\documentclass[11pt]{article}

\usepackage{changepage}
\usepackage{enumitem}

\usepackage[usenames,dvipsnames,svgnames,table]{xcolor}
\usepackage[titletoc,title]{appendix}
\usepackage{amsfonts}
\usepackage{graphicx}
\usepackage{hyperref}
\usepackage{amsmath}
\usepackage{amssymb}
\usepackage{framed}
\usepackage{amsthm}
\usepackage{bbm}
\usepackage{comment}
\usepackage{cancel}
\usepackage{float}
\usepackage{listings}

\newtheorem{theorem}{Theorem}[section]

\newtheorem{lemma}[theorem]{Lemma}
\newtheorem{claim}[theorem]{Claim}
\newtheorem{corollary}[theorem]{Corollary}

\newtheorem{definition}[theorem]{Definition}

\newtheorem{conjecture}[theorem]{Conjecture}

\theoremstyle{definition}
\newtheorem{remark}[theorem]{Remark}

\theoremstyle{definition}
\newtheorem*{remark*}{Remark}

\newtheorem*{proposition*}{Proposition}

\theoremstyle{definition}

\newcommand{\mrm}[1]{\mathrm{#1}}
\newcommand{\skipline}{$\phantom{}$}

\newcommand{\anderbrace}[2]{
	\if\relax\detokenize{#2}\relax
	\sbox0{$\underbrace{#1}_{}$}
	\mathrel{\mathmakebox[\wd0]{#1}}
	\else
	\mathrel{\underbrace{#1}_{\mathclap{#2}}}
	\fi}

\newcommand{\integers}{\mathbb{Z}}

\newcommand{\reals}{\mathbb{R}}

\newcommand{\field}{\mathbb{F}}
\newcommand{\spm}{\left\{-1,1\right\}}
\newcommand{\zo}{\left\{0,1\right\}}
\newcommand{\czo}{\left[0,1\right]}

\newcommand{\seq}{\subseteq}
\newcommand{\sm}{\setminus}

\newcommand{\es}{\emptyset}
\newcommand{\defeq}{\mathrel{\mathop:}=}

\newcommand{\im}{\mrm{Im}}

\newcommand{\one}{\mathbbm{1}}

\newcommand{\pr}{\Pr}
\newcommand{\given}[2]{#1\,\middle|\, #2}

\DeclareMathOperator*{\be}{\mathbb{E}}

\DeclareMathOperator*{\cov}{\mrm{Cov}}

\newcommand{\andd}{\wedge}

\newcommand{\wh}[1]{\widehat{#1}}
\newcommand{\wt}[1]{\widetilde{#1}}
\newcommand{\norm}[1]{\left\Vert {#1}\right\Vert}

\newcommand{\li}{\left}
\newcommand{\ri}{\right}

\newcommand{\cc}{\colon}
\newcommand{\func}[3]{{#1}\cc {#2} \to {#3}}

\newcommand{\set}[2]{\left\{ \given{#1}{#2}\right\}}

\newcommand{\restrict}[2]{{
		\left.\kern-\nulldelimiterspace 
		#1 
		\vphantom{\big|} 
		\right|_{#2} 
}}

\usepackage{titlesec}
	
\titlespacing{\paragraph}{%
		0pt}{
		0.5\baselineskip}{
		1em}
	
\newcommand{\tops}[1]{\texorpdfstring{#1}{}}

\newcommand{\Z}{\mathbb{Z}}
\newcommand{\z}{z}
\newcommand{\complex}{\mathbb{C}}
\newcommand{\modd}{\,\%\,}
\newtheorem*{lemma*}{Lemma}

\newtheorem{alg}[theorem]{Algorithm}

\begin{document}
	\title{Fine-Grained Cryptanalysis:\\ Tight Conditional Bounds for Dense $k$-SUM and $k$-XOR}
\author{
	Itai Dinur\footnote{The first author was supported by the Israel Science Foundation (grants no. 573/16 and 1903/20).} \\\small Ben-Gurion University \\\small dinuri@cs.bgu.ac.il
    \and
	Nathan Keller\footnote{This research was supported by the European Research Council under the ERC starting grant agreement no. 757731 (LightCrypt) and by the BIU Center for Research in Applied Cryptography and Cyber Security in conjunction with the Israel National Cyber Bureau in the Prime Minister’s Office. In addition, the second author was supported by the Israel Science Foundation (grants no. 1612/17 and 2669/21).} \\\small Bar-Ilan University \\\small Nathan.Keller@biu.ac.il
	\and
	Ohad Klein\footnote{The work was performed when the third author was in Bar-Ilan University and was supported by the Clore Scholarship Programme.} \\\small Hebrew University \\\small ohadkel@gmail.com
}
    \date{}
	\maketitle
	
	\begin{abstract}
An average-case variant of the $k$-SUM conjecture
asserts that finding $k$ numbers that sum to 0 in a list of $r$ random numbers,
each of the order $r^k$, cannot be done in much less than $r^{\lceil k/2 \rceil}$ time.
On the other hand, in the \emph{dense regime} of parameters,
where the list contains more numbers and many solutions exist,
the complexity of finding one of them can be significantly improved by Wagner's $k$-tree algorithm.
Such algorithms for $k$-SUM in the dense regime have many applications, notably in cryptanalysis.

In this paper, assuming the average-case $k$-SUM conjecture,
we prove that known algorithms are essentially optimal for $k= 3,4,5$.
For $k>5$, we prove the optimality of the $k$-tree algorithm for a limited range of parameters.
We also prove similar results for $k$-XOR, where the sum is replaced with exclusive or.

Our results are obtained by a self-reduction that, given an instance of $k$-SUM which has a few solutions, produces from it many instances in the dense regime.
We solve each of these instances using the dense $k$-SUM oracle,
and hope that a solution to a dense instance also solves the original problem.
We deal with potentially malicious oracles (that repeatedly output correlated useless solutions)
by an obfuscation process that adds noise to the dense instances.
Using discrete Fourier analysis, we show that the obfuscation eliminates
correlations among the oracle's solutions, even though its inputs are highly correlated.
	\end{abstract}


\pagenumbering{Alph}
\thispagestyle{empty}
\newpage
\pagenumbering{arabic}
\setcounter{page}{1}
	
\section{Introduction}

\subsection{Background}

\paragraph{The $k$-SUM problem.} For parameters $k = O(1),r$, the classical worst-case search variant of the $k$-SUM problem asks: Given a list of $r$ numbers, find (with high probability) $k$ of them whose SUM is zero, assuming such numbers exist.\footnote{\label{foot:lists} Another variant of the $k$-SUM problem asks, given $k$ lists of $r/k$ numbers, find $k$ numbers -- one from each list -- whose SUM is zero. The two problems are equivalent, up to $O_k(1)$ factors.}
Given that a solution exists, a simple sort-and-match (or meet-in-the-middle) algorithm finds it in time $T=\tilde{O}(r^{\lceil k/2 \rceil})$
(the notation $\tilde{O}$ hides logarithmic factors in $r$),
and the well-known $k$-SUM conjecture (that generalizes the $3$-SUM conjecture~\cite{GO95}),
states that no algorithm can do substantially better in standard computational models (such as the word RAM model).

In this paper, we consider average-case variants of the $k$-SUM problem.
\begin{definition}[Average-case $k$-SUM problem]
In the $(k,N,r)$-SUM problem, the input consists of $r$ elements $\z_1,\ldots,\z_r$, each of them chosen independently and uniformly at random from $\{-N,\ldots,N\}$. The goal is to find a $k$-tuple (an ordered set of distinct indices) $K = \{i_1,\ldots,i_k\}$, such that $\sum_{j \in K} \z_j = 0$, where the sum is over $\mathbb{Z}$.
\end{definition}
In a \emph{sparse} regime of parameters \emph{only a few solutions exist on average}, i.e., $r^k \approx N$.
It is considered folklore that the uniform distribution is a hard distribution for $k$-SUM under a standard model of computation
(see~\cite{LLW19} and~\cite{P15} for a formulation for $k=3$):

\begin{conjecture}[Sparse average-case $k$-SUM conjecture]
\label{conj:ksum}
Any algorithm that solves the $(k,N,r)$-SUM problem where $r = N^{1/k}$ with probability $\Omega_k(1)$ has expected running time of at least $T= \Omega_k( r^{\lceil k/2 \rceil- o(1)})$.
\end{conjecture}

We note that for constant $k$ and $N = \omega(r^k)$, a solution exists with probability $o(1)$, hence the problem cannot be solved with probability $\Omega_k(1)$, regardless of the running time.

In the \emph{dense} regime where \emph{many solutions exist} on average (namely, when $r^k \gg N$), one can do much better.

For $k=3$, there is a simple algorithm that filters the input by keeping only numbers that are smaller than some threshold in absolute value. This gives a smaller sparse instance to which the standard algorithm is applied to solve the problem in time
$T = \tilde{O}(N/r)$ (for $N^{1/3} \leq r \leq N^{1/2}$). For $k > 3$, improvements are obtained via the celebrated Wagner's \emph{$k$-tree algorithm}~\cite{W02} discussed below.

\paragraph{The $k$-XOR problem.} The discussion above equally applies to the average-case $k$-XOR problem.
\begin{definition}[Average-case $k$-XOR problem]
\label{def:kxor}
In the $(k,2^n,r)$-XOR problem, the input consists of $r$ vectors $\z_1,\ldots,\z_r$, each chosen independently and uniformly at random from $\{0,1\}^n$. The goal is to find a $k$-tuple, $K = \{ i_1,\ldots,i_k\}$, such that $\bigoplus_{j \in K} \z_j = 0_n$.
\end{definition}
Similarly to $k$-SUM, the following conjecture is considered folklore.
\begin{conjecture}[Sparse average-case $k$-XOR conjecture]
\label{conj:kxor}
Any algorithm that solves the $(k,N = 2^n,r)$-XOR problem where $r = N^{1/k}$ with probability $\Omega_k(1)$ has expected running time of at least $T= \Omega_k( r^{\lceil k/2 \rceil- o(1)})$.
\end{conjecture}
In the dense $3$-XOR problem, the input consists of $r \gg 2^{n/3}$ uniform vectors.
Similarly to $3$-SUM, a simple filtering algorithm has complexity $T = \tilde{O}(N/r)$ for $N^{1/3} \leq r \leq N^{1/2}$.
The dense $3$-XOR problem has various applications in cryptography and cryptanalysis~\cite{BDF18,DKMN21,J09,LS19,NS15}.
While mild (logarithmic in $N$) improvements to the simple filtering algorithm are known~\cite{J09,LS19,NS15}, any substantial (i.e., polynomial in $N$) improvement would be considered a breakthrough.

\paragraph{Wagner's $k$-tree algorithm.} For $k > 3$, the \emph{$k$-tree algorithm} of Wagner~\cite{W02}
allows finding a solution to $k$-XOR in time $T= \tilde{O}(N^{1/(1+\lfloor \log_2 k \rfloor)})$, when $r$ is of similar size.
The generalized algorithm of Minder and Sinclair~\cite{MS12} provides a tradeoff between $r$ and $T$, for all
$N^{1/k} \leq r \leq N^{1/(1+\lfloor \log_2 k \rfloor)}$. For the most basic case of $k=4$, the tradeoff curve is $T=\tilde{O}(N/r^2)$ for $N^{1/4} \leq  r \leq N^{1/3}$. As was noted in~\cite{W02}, the algorithm is also applicable to the modular $k$-SUM problem in $\integers_N$. Similarly, it can be easily modified to work for the average-case variant of $k$-SUM stated above. For the sake of completeness, we give a high-level overview of the $k$-tree algorithm and its generalization in Appendix~\ref{app:ktree}.

In the 20 years since its introduction, the $k$-tree algorithm (notably for small $k$ values) has become a central tool in cryptanalysis for solving both dense $k$-SUM and $k$-XOR problems (see~\cite{J09}). Specifically, it is used in breaking hash functions~\cite{MPRKS08,W02}, stream ciphers~\cite{LV04}, block ciphers~\cite{DKS15}, signature schemes~\cite{BenhamoudaLLO021} (where the optimal value of $k$ depends on the amount of available data), etc. Furthermore, it has found multiple applications that are not directly related to cryptanalysis. Notably, the \emph{representation technique}~\cite{HGJ10}
crucially relies on variants of the algorithm for small values of $k$ to find one out of many representations of a solution to a problem. This technique gave rise to breakthrough algorithms for solving subset-sum~\cite{HGJ10,NW21} and related problems such as decoding binary linear codes~\cite{BJMM12}.

Finally, the $k$-tree algorithm is closely related to the Blum-Kalai-Wasserman (BKW) algorithm for solving the LPN (learning parity with noise) problem~\cite{BKW03} and its extensions, such as Lyubashevsky's algorithm~\cite{L05} (although these use $k = \omega(1)$).

In this paper we address the question: \textbf{Are the best-known algorithms for dense $\mathbf{k}$-SUM and $\mathbf{k}$-XOR optimal?}

\subsection{Our results}

We show that in some of the most basic cases, $k=3,4,5$, as well as in other settings, the best known algorithms for $k$-SUM (resp., $k$-XOR) in the dense regime are optimal up to logarithmic factors in the input list size, unless the sparse average-case $k$-SUM (resp., $k$-XOR) conjecture fails.

\paragraph{Informal statement of the main results.}

Our main theorem for $k$-SUM is as follows.
\begin{theorem}[Conditional dense $k$-SUM hardness, informal]\label{thm:intro-main-informal-SUM}
Assume that any algorithm that solves $(k,N,N^{1/k})$-SUM with probability $\Omega_k(1)/(\log N)^2$ has expected running time of at least $T = T(N,k)$.

Then, there is $C=C(k)$ such that for any $0 \leq \epsilon \leq 1/2$, any algorithm that solves $(k,N',(N')^{(1 + \epsilon)/k})$-SUM with probability $1/2$ has expected running time of at least $C \cdot T((N')^{1 + \epsilon}, k) \cdot (N')^{- \epsilon}$.
\end{theorem}

\begin{remark}
We make several related remarks about the theorem.
\begin{itemize}
	\item While the success probability in the hardness assumption of Theorem~\ref{thm:intro-main-informal-SUM} is slightly smaller than the constant success probability in Conjecture~\ref{conj:ksum}, disproving this stronger assumption with the same time complexity, would be considered a breakthrough (e.g., for cryptanalytic applications).
	
	\item It is possible to amplify the success probability in the hardness assumption from $\Omega_k(1)/(\log N)^2$ to $1/2$ (or any constant), at the cost of increasing the number of input elements $r$ and the expected running time of the algorithm by a factor of $(\log N)^2 \cdot O_k(1)$. The amplification is obtained by partitioning the elements into disjoint groups of size $N^{1/k}$,
	and running the algorithm for $(k,N,N^{1/k})$-SUM on each group independently.
	
	\item The input size of $N^{1/k}$ in the conditional hardness assumption can be adjusted to $C_1 \cdot N^{1/k}$ for any constant $C_1 > 0$. This only requires adjusting the hidden constant behind the success probability $\Omega_k(1)/(\log N)^2$.
	\end{itemize}
\end{remark}

We obtain similar results for $k$-XOR
(with constant success probability in the hardness assumption, as in Conjecture~\ref{conj:kxor}).
Our main theorem for $k$-XOR is as follows.

\label{sec:main}
\begin{theorem}[Conditional dense $k$-XOR hardness, informal]\label{thm:intro-main-informal-XOR}
Assume that any algorithm that solves $(k,N,N^{1/k})$-XOR with probability $\Omega_k(1)$ has expected running time of at least $T = T(N,k)$.

Then, there is $C=C(k)$ such that for any $0 \leq \epsilon \leq 1/2$, any algorithm that solves $(k,N',(N')^{(1 + \epsilon)/k})$-XOR with probability $1/2$ has expected running time of at least $C \cdot T((N')^{1 + \epsilon}, k) \cdot (N')^{- \epsilon}$.
\end{theorem}

\paragraph{Discussion.}
To better understand the tradeoff obtained by the theorem, set $N' = N$ and $T(N,k) = \Omega(N^{\alpha(k) -o(1)})$ (for some function $\alpha(k)$). The theorem implies that any algorithm for $(k,N,r)$-SUM with $r \approx N^{(1 + \epsilon)/k}$ (for $0 \leq \epsilon \leq 1/2$) that succeeds with probability $1/2$ has expected running time of
\begin{align*}
\Omega( N^{(\alpha(k) - o(1))(1 + \epsilon) - \epsilon})
&=  \Omega(N^{1 - o(1)} \cdot N^{-1 + \alpha(k) + \epsilon \cdot \alpha(k) - \epsilon}) \\
&= \Omega(N^{1 - o(1)} \cdot N^{(1+\epsilon)(\alpha(k) - 1)}) \\
&= \Omega(N^{1 - o(1)} \cdot r^{k \cdot(\alpha(k) - 1)}).
\end{align*}
Assuming Conjecture~\ref{conj:ksum}, we plug-in $\alpha(k) = \lceil k/2 \rceil / k$, and derive the conditional lower bound
\[
T = \Omega(N^{1 - o(1)} / r^{\lfloor k/2 \rfloor}).
\]
Specifically, for $k=3$, by a slightly stronger variant of the sparse average-case $k$-SUM conjecture (with the success probability in the hardness assumption adjusted to $\Omega_k(1)/(\log N)^2$),
we deduce that any algorithm for $(3,N,r)$-SUM
that succeeds with probability $1/2$ has expected running time of
$\Omega({N^{1-o(1)}}/r)$.
We conclude that
the tradeoff $T = \tilde{O}(N/r)$ for $N^{1/3} \leq r \leq N^{1/2}$
obtained by the simple (filtering) sort-and-match algorithm is essentially optimal.
Similar tightness holds for the $3$-XOR problem.
	
By a similar calculation for $k=4$, the (extended) $k$-tree algorithm (obtaining the tradeoff $T=\tilde{O}(N/r^2)$) is essentially optimal, under the sparse average-case $k$-SUM (resp., $k$-XOR) conjecture.
Theorems~\ref{thm:intro-main-informal-SUM} and~\ref{thm:intro-main-informal-XOR} yield optimality of the best known algorithm also for $k=5$ and for part of the range for other values of $k$. In particular, for even values of $k$, we conclude that
any algorithm for $(k,N,r)$-SUM (resp., XOR)
that succeeds with probability $1/2$ has expected running time of
$T = \Omega(N^{1 - o(1)} / r^{k/2})$.
This essentially matches the extended $k$-tree algorithm for $k$ values divisible by 4 in the range $N^{1/k} \leq r \leq N^{4/3k}$.

We further note that the loss in the reduction provided by the above theorems is almost \emph{linear}, i.e., a $O_k(1)$ factor for $k$-XOR and logarithmic in $N$ for $k$-SUM.
This means that for $k = 3,4,5$ any improvement of known algorithms, \emph{even by a sufficiently large logarithmic factor in $N$}, can be leveraged through the theorem to obtain a similar improvement in the algorithms for the sparse average-case $k$-SUM (resp., $k$-XOR) problem.

Figure~\ref{fig:tradeoff} shows our $k$-SUM and $k$-XOR density-complexity tradeoff lower bounds compared to the best-known upper bounds for $k = 3,4,8$ (ignoring logarithmic factors).

\begin{figure}[ht]
 \begin{center}
 \includegraphics[scale=0.8,clip=true]{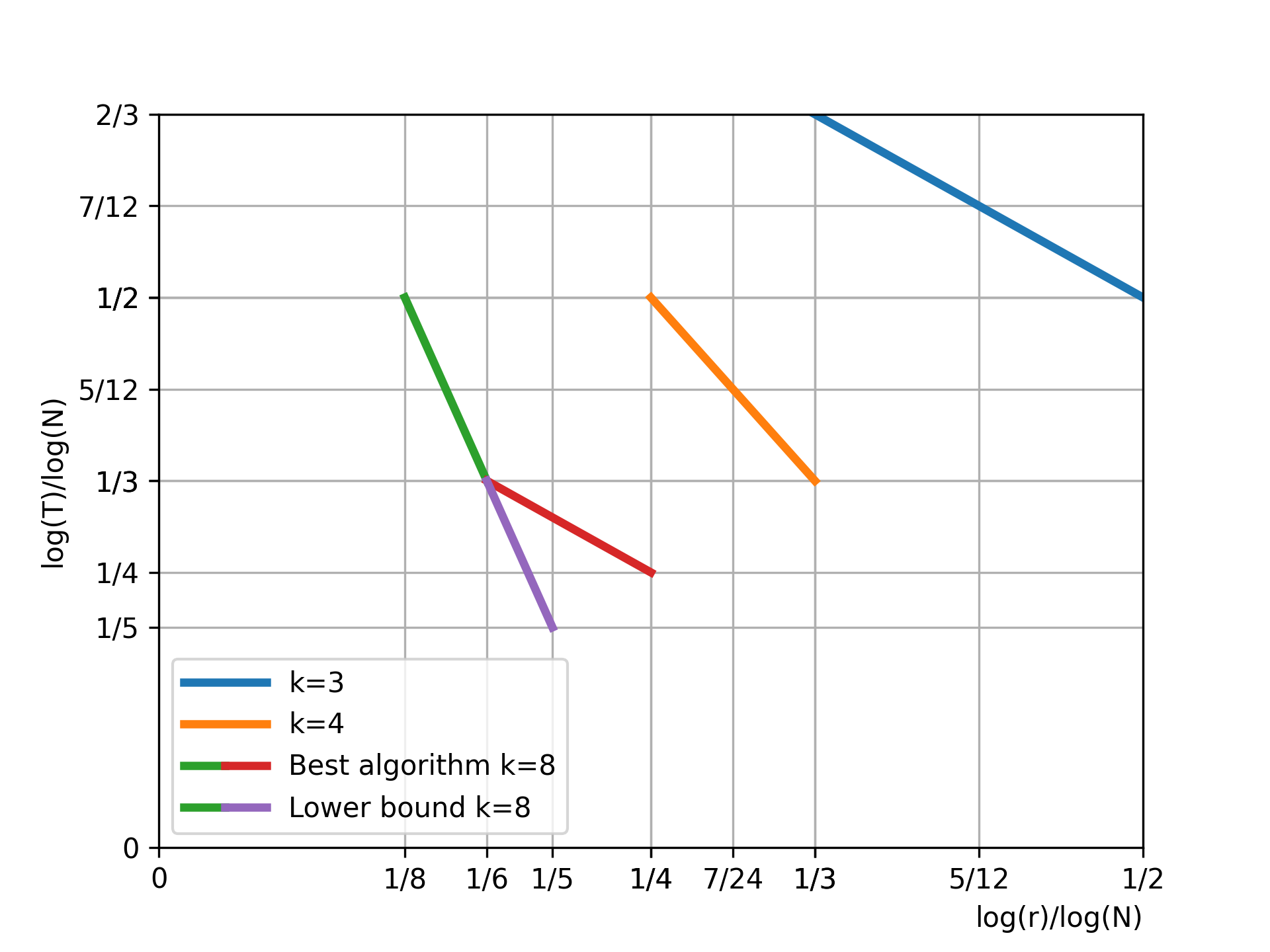}
 \caption{$k$-SUM and $k$-XOR density-complexity tradeoff lower and upper bounds for $k = 3,4,8$}
 \label{fig:tradeoff}
 \end{center}
\end{figure}

\subsection{Our methods}

We achieve our results by a self-reduction from the average-case $k$-SUM ($k$-XOR) problem in the sparse regime, to the average-case $k$-SUM ($k$-XOR) problem in the dense regime.
In the following description, we focus on $k$-XOR, as technical details are simpler for this problem. We then summarize the main different ingredients for $k$-SUM.

\paragraph{The reduction for $k$-XOR.}The basic observation is that we can generate a dense instance with $r$ input vectors from a sparse instance with $r$ input vectors by truncating the $n$-bit input vectors to obtain shorter $m$-bit input vectors for $m < n$. This increases the effective input list size (relative to the vector length) and the number of solutions.

We rewrite our main result in a more convenient form based on this observation by a change of variables for the dense regime:
\begin{theorem}[Conditional dense $k$-XOR hardness, informal, reformulated]\label{thm:intro-main-informal-XOR-2}
Assume that any algorithm that solves $(k,N=2^n,N^{1/k} = 2^{n/k})$-XOR with probability $\Omega_k(1)$ has expected running time of at least $T = T(N,k)$.

Then, there is $C=C(k)$ such that for any $n/2 \leq m \leq n$,
any algorithm that solves $(k,M = 2^m, 2^{n/k})$-XOR with probability $1/2$
has expected running time of at least $C \cdot T(2^m,k) \cdot 2^{m - n}$.		
\end{theorem}
Observe that this is a indeed a reformulation of~Theorem~\ref{thm:intro-main-informal-XOR}, obtained by setting $n'=m$ and $n = (1 + \epsilon)m$
(hence, $(N')^{- \epsilon} = 2^{- \epsilon n' } = 2^{- \epsilon m} = 2^{m - n}$).

The basic idea of the reduction is to take the sparse input of $r$ uniform $n$-bit vectors, generate from it \emph{many dense $k$-XOR inputs} of $r$ uniform $m$-bit vectors, and solve each one using a black-box algorithm $B$. We then check whether each solution yields a $k$-XOR solution for the original $n$-bit vectors. If we could make the input sets of $m$-bit vectors \emph{completely independent}, then $O(2^{n-m})$ calls to $B$ would be sufficient (as the probability that a solution for $m$-bit vectors corresponds to a solution for $n$-bit vectors is $2^{m-n}$),
and the assertion of Theorem~\ref{thm:intro-main-informal-XOR-2} would be achieved. In fact, it is easy to show that pairwise independence suffices. However, one cannot make these input sets pairwise independent and maintain their relation to the original $n$-bit vectors at the same time (unless $m \leq n/2$, which is not useful in our case). Thus, even though it is called about $2^{n-m}$ times, $B$ could potentially repeatedly output solutions to dense $k$-XOR instances that reside in a small set which does not contain any solution to the sparse instance.

We overcome this obstacle by an \emph{obfuscation} process, which applies to the input vectors two different types of noise consecutively, and allows us to achieve almost pairwise independence of \emph{$B$'s outputs} even though its inputs are \emph{significantly correlated}.
In the reduction, we are given a sequence $\z_1, \ldots, \z_r \in \{0,1\}^n$ for which we wish to solve the $k$-XOR problem. We apply the following procedure.
\begin{enumerate}
			\item \textbf{Draw} a uniformly random matrix $T \in \{0,1\}^{m \times n}$ of full rank $m$, and a uniformly random permutation $P$ on $r$ elements.		
			\item \textbf{Let} $x_{i} = T(\z_{P(i)})$ for all $i \in [r]$.	
			\item \textbf{Feed} $B$ with $x_1, \ldots, x_r$. In case it outputs a $k$-tuple $K$ with $\bigoplus_{i \in K} x_i = 0_m$, \textbf{test} whether
			$\mathcal{K} = P(K)$ satisfies $\bigoplus_{j \in \mathcal{K}} \z_{j} = 0_n$, and \textbf{if} so -- \textbf{output} the $k$-tuple $\mathcal{K}$. \textbf{Otherwise, repeat}.
\end{enumerate}

We prove that after $2^{n-m}$ trials, with probability of $\Omega_k(1)$, the process outputs a solution of the sparse $k$-XOR problem. We use \emph{discrete Fourier analysis} in order to bound the correlation between $B$'s outputs.

\begin{remark}
If we hash down a hard sparse instance to get a dense instance, then clearly any procedure that enumerates all solutions to the dense instance (if there aren't too many) is hard as well.
However, our reduction does not follow this standard paradigm, as the oracle for the dense instance can only produce a single solution. Thus, we hash the sparse instance down in many different ways and repeatedly invoke the oracle in order to force it to produce many different potential solutions to the sparse instance.
\end{remark}

\paragraph{The reduction for $k$-SUM.} The reduction for $k$-SUM follows the same general strategy (with modified obfuscation), but its proof is more involved. In particular, in addition to discrete Fourier-analytic techniques, it uses tools from Littlewood-Offord theory~\cite{Erdos45}.

Concisely, the reason for further complexity in the $k$-SUM case, is that there are only a few group-homomorphisms from $\integers_{N}$ to $\integers_{M}$ (that may be used to obfuscate the input), while group-homomorphisms from $\field_2^n$ to $\field_2^m$ are abundant. We now elaborate on this point.

Both $k$-XOR and $k$-SUM reductions employ an obfuscation procedure to the input of the sparse problem. For $k$-XOR we hash $r$ elements of $\zo^n$ to $r$ elements of $\zo^m$ ($m < n$), while for $k$-SUM we hash $r$ elements of $\integers_{N}$ to $r$ elements of $\integers_{M}$. Since we are required to pull-back linear information from the output of the process (i.e. a solution to a dense problem), to its input (i.e., solve the sparse problem), it is important that the obfuscation would be linear. However, the class of \emph{obfuscation} functions must be rich enough to mask additional information that can be exploited by the dense algorithm to output correlated solutions. In the $k$-XOR case, we make use of a random, rank-$m$, linear map from $\field_2^n$ to $\field_2^m$. There are $2^{\Theta(mn)}$ options to choose this linear map.  In the $k$-SUM problem, if we would insist on a (surjective) linear map from $\integers_N$ to $\integers_M$, there would be at most $M$ such functions, which gives insufficient obfuscation (in particular, any $x \in \integers_N$ that is divisible by $M$, would be mapped to $0$). Hence, we must settle for a \emph{somewhat-linear} obfuscation. We choose a function $\func{\phi}{\integers_N}{\integers_M}$ of the form
\[
	\phi(x) = \gamma \lfloor \alpha \cdot x \cdot M/N \rceil \bmod M,
\]
where $\alpha \in \integers_{N}^\ast$, $\gamma \in \integers_{M}^\ast$ are chosen uniformly at random and $\lfloor \cdot \rceil$ denotes rounding to the nearest integer. This new class of non-strictly-linear obfuscation functions still enables us to pull back a solution from a dense problem to the sparse problem. Moreover, it turns out that this class is rich enough to allow obfuscation. However, proving this latter point is more involved, since the non-linearity of the obfuscation makes the use of Fourier-analytic tools more complex.

\subsection{Related work}

Our result is naturally related to three lines of work.

\paragraph{Algorithms for dense $k$-SUM and $k$-XOR.} The first line of work is the quest for designing better algorithms for \emph{generalized birthday problems}, where the goal is to find a single $k$-SUM or $k$-XOR solution out of many. The systematic analysis of this problem was initiated by Wagner~\cite{W02} in his $k$-tree algorithm,\footnote{We note that variants of the $k$-tree algorithm were presented earlier, starting with Camion and Patarin~\cite{CP91}, as is stipulated in~\cite[full version]{W02}.} and has led to numerous refined algorithms and applications thereof~\cite{Dinur19,J09,LS19,MS12,NS15}.

In this respect, we show that -- perhaps surprisingly -- the best-known algorithms are \emph{essentially optimal} for $k=3,4,5$, unless a standard conjecture in computational complexity fails. Moreover, under similar conjectures, the best known algorithms for any $k > 5$ are essentially optimal for some range of parameters.

\paragraph{Fine-grained complexity.} The second related line of work deals with complexity reductions to the $k$-SUM problem and its variants, that have become a flourishing field of research in the last decade, as part of the \emph{fine-grained complexity} research direction~\cite{VW19}. Reductions to $k$-SUM or to $k$-XOR were shown for problems in computational geometry~\cite{AH08,GO95}, dynamic algorithms~\cite{AW14,Patrascu10}, graph algorithms~\cite{JV16,WW13}, pattern matching~\cite{AWW14} and more. In the last few years, such reductions were shown also for several cryptographic problems~\cite{BRSV18,GGHPV20,LLW19}, as part of the emerging \emph{fine-grained cryptography} research area~\cite{DVV16}.

Our work provides yet another reduction for a cryptography-related problem, however the context of reduction in our case is somewhat different. While previous works prove security of (mainly theoretical) \emph{classes of cryptographic primitives}, based on well-founded hardness assumptions, our work shows \emph{a bound on the possible effectiveness of an important class of cryptanalytic algorithms} that are widely used for breaking cryptosystems. Thus, our results also have more practical significance.

\paragraph{Asymptotic hardness for dense $k$-SUM.} Finally, a recent work related to ours is the paper~\cite{BDV20} by Brakerski, Stephens-Davidowitz, and Vaikuntanathan, which proves asymptotic optimality of the $k$-tree algorithm for average-case $k$-SUM (whose complexity is $N^{O(1/\log k)}$) by reducing it from \emph{worst-case complexity of lattice problems}. While~\cite{BDV20} is related to our work, the results of the two papers are complementary due to several important differences which we summarize below.

First,~\cite{BDV20} yields asymptotic bounds as $k \to \infty$, while our work concentrates on small values of $k$, which appear in applications of the $k$-tree algorithm to a different type of cryptanalytic problems. In this respect, our result resolves an open problem stated in~\cite[Sec.~1.3]{BDV20}, yet our reduction is not from a worst-case problem (due to various technical differences, it is not clear how to combine our techniques with the ones of~\cite{BDV20} to obtain a reduction from a worst-case problem). Second, while the bounds of~\cite{BDV20} are tight up to a constant multiplicative factor in the exponent, our reduction is tight up to constant (or logarithmic) factors for certain parameter ranges such as $k=3,4,5$. Third, while the reduction of~\cite{BDV20} is from a different problem, involving lattices, our reduction is from a conjecture in the sparse regime to the dense regime of the same problem (i.e., a \emph{self-reduction}). We note that the density (or size) of the instance also plays a role in~\cite{BDV20}, as a faster algorithm for sparser $k$-SUM instances yields a faster algorithm for the corresponding lattice problem. Fourth, in addition to $k$-SUM, we also obtain conditional hardness results for the $k$-XOR problem.

In terms of techniques, both papers aim at obtaining `sufficiently different' variants of the same input sample $z$. However, in~\cite{BDV20} these are obtained by \emph{re-randomization} (repeatedly generating almost independent inputs to the $k$-SUM algorithm from the same list of vectors)
and their independence is proved via the \emph{leftover hash lemma}~\cite{ILL89}.
On the other hand,
our variants are obtained via the \emph{obfuscation} method described above
in which the inputs to the $k$-SUM algorithm are highly correlated,
unlike the setting of~\cite{BDV20}.
We prove the low correlation of the algorithm's outputs via \emph{discrete Fourier analytic} methods.

\subsection{Additional application and open problems}

\paragraph{Additional application.}
The security proof of the hash construction $T5$, recently proposed by Dodis et al.~\cite{DKMN21},
is based on dense $3$-XOR and $4$-XOR assumptions.
Our results directly imply that the security of the construction can be based on standard sparse $3$-XOR and $4$-XOR assumptions instead of non-standard dense ones. In this sense, our work (in combination with the original security proof of~\cite{DKMN21}) allows to prove security for a cryptosystem, similarly to~\cite{BRSV18,GGHPV20,LLW19}, yet this proof is obtained for a practical cryptosystem.

\paragraph{Open problems.} The main remaining open problem is to improve our lower bound in the setting of a large $k$ and a large number of solutions, or alternatively, to improve the $k$-tree algorithm in this range.

\paragraph{The structure of the paper.} Next, we summarize our notations and conventions.
In Section~\ref{sec:kxor} we prove our main result for $k$-XOR,
while in Section~\ref{sec:k-SUM} we prove our main result for $k$-SUM.

\section{Notations and Conventions}
\label{sec:notations}

In this section we introduce notations and conventions that will be used throughout the paper.

\paragraph{Notations.}

\begin{itemize}[noitemsep]
	\item $x \sim S$ means that $x$ is a random variable uniformly distributed in the set $S$.

	\item We interchangeably write $\zo$ and $\field_2$, where $0=1 \oplus 1$.
	
	\item When $\z \in \field_2^{r \times n}$, $\z_{ij}$ denotes the $(i,j)$'th entry of $\z$, and $\z_i\defeq (\z_{i1}, \ldots, \z_{in})$.
	
	\item For $x \in \field_2^{r \times m}$ and a permutation $P \in S_r$, we denote by $P(x)$ the value $y \in \field_2^{r \times m}$ satisfying $y_{P(i)}=x_i$ for all $i$.
	
	\item For a linear map $\func{T}{\field_2^n}{\field_2^m}$ and for $\z \in \field_2^{r \times n}$, we denote by $T(\z)$ the value $x \in \field_2^{r \times m}$ satisfying $x_{i}=T(\z_i)$ for all $i$.
	
	\item $\mathbb{Z}_L$ is the group whose elements are $\{0,1,\ldots,L-1\}$, and whose operation is addition modulo $L$.
	
	\item $a \modd b$ (or $a \bmod b$) stands for the single element in $(a + b\integers) \cap [0, b)$. We switch to this shorter notation in subsection~\ref{sec:sub:sum-obfuscation-proof}, as it is more convenient to use inside long arithmetic expressions.
	
	\item For a real number $u$, the rounded value $\lfloor u\rceil$ is the unique integer in $u+(-1/2, 1/2]$.
	
	\item For functions $f,g: \mathbb{N} \to \mathbb{R}_+$ and a fixed parameter $k$, $f=O(g)$ means that $\forall n: f(n) \leq Cg(n)$ for an absolute constant $C$, $f=\tilde{O}(g)$ means that $\forall n: f(n) \leq C_1 \cdot (\log n)^{C_2} \cdot g(n)$ for absolute constants $C_1,C_2$, and $f=O_k(g)$ means that $\forall n: f(n) \leq h(k) \cdot g(n)$ for some function $h:\mathbb{N} \to \mathbb{R}_+$.
\end{itemize}

\paragraph{Conventions.}

\begin{itemize}[noitemsep]
	\item \textbf{Operations within domains.} Throughout the paper, we consider variables in various domains. For example, when analyzing $k$-XOR, we consider variables in $\{0,1\}^{\ell}$ and $\{0,1\}^{r \times \ell}$ for different values of $\ell,r$, while for $k$-SUM, we consider variables in $\Z_L$ and $\Z_L^r$ for different values of $L,r$. Whenever an operation is applied on two elements of the same domain, the result belongs to the same domain. For example, addition between two elements of $\Z_L$ is always performed modulo $L$.
	
	\item \textbf{Names of variables.} In the reductions presented in the paper, we begin with vectors that belong to a `large' space -- $\{0,1\}^{r \times n}$ for $k$-XOR ($\Z_N^r$ for $k$-SUM), and use them to define vectors that belong to a `smaller' space -- $\{0,1\}^{r \times m}$ for $k$-XOR ($\Z_M^r$ for $k$-SUM).
	
	Throughout the paper, vectors that belong to a `large' space are denoted by $\z$, while vectors that belong to a `small' space are denoted by $x$ or $y$. Auxiliary vectors denoted by $u$ or $v$ may belong to an arbitrary domain, which will be explicitly defined.
	
	\item \textbf{Inner products.} All inner products in Section~\ref{sec:kxor} are between functions $f,g:\{0,1\}^{r \times m} \rightarrow \mathbb{R}$ (for a particular choice of $m$), and consequently, their results belong to $\mathbb{R}$. Most inner products in Section~\ref{sec:k-SUM} are between vectors in $\Z_p^r$ (for a particular choice of $p$), and consequently, their results belong to $\Z_p$.
\end{itemize}


\section{Hardness of Dense \tops{$k$}-XOR}
\label{sec:kxor}

In this section we prove Theorem~\ref{thm:intro-main-informal-XOR} (or equivalently, Theorem~\ref{thm:intro-main-informal-XOR-2}).
The precise formulation of the main theorem is as follows.
\begin{theorem}[Sparse to dense $k$-XOR reduction]
\label{thm:red}
        Let $m,n$ be integers such that $n/2 \leq m \leq n$.
        Assume there is an algorithm for $(k,2^m, 2^{n/k})$-XOR with
        success probability $\beta$ and expected running time $\mathcal{T}$.
		Then, there is an algorithm for $(k, 2^n, 2^{n/k})$-XOR
		with success probability $\frac{\beta^4}{(16k)^{k+2}}$, and expected running time at most
		$2^{n-m} \cdot (\mathcal{T}+ \tilde{O}_k(2^{n/k}))$.
\end{theorem}

We make a few remarks.
First, we work in the standard word RAM computational model, where an operation on each vector (of size $O(n)$) takes unit time.
However, our results do not change substantially in other standard computational models (e.g., in a model that counts the number of
bit operations).
Second, when applying the contrapositive of this theorem to prove hardness results for dense $k$-XOR with $k=O(1)$ and $\beta = 1/2$ to obtain the theorems~\ref{thm:intro-main-informal-XOR} and~\ref{thm:intro-main-informal-XOR-2},
the factor $\tilde{O}_k(2^{n/k})$ in this theorem is bounded by $O_k(\mathcal{T})$ (and hence consumed by the factor $C = C(k)$ of theorems~\ref{thm:intro-main-informal-XOR} and~\ref{thm:intro-main-informal-XOR-2}). Indeed, the input in the $(k,M = 2^m, 2^{n/k})$-XOR problem already contains $2^{n/k}$ vectors. Finally, the loss factor $\tfrac{\beta^4}{(16k)^{k+2}}$ in the success probability can be significantly improved by a more refined analysis.

The proof of the theorem is based on two reductions: the main sparse to dense $k$-XOR reduction
and the simple (but inefficient) sparse to dense reduction which is required for parameter ranges where the main reduction is not applicable.
We summarize these reductions in the two lemmas below and prove the simpler lemma first.
\begin{lemma}[Main $k$-XOR reduction]
\label{lem:mainxor}
		Let $m,n,r$ satisfy $n/2 \leq m \leq n \leq \log_2 \binom{r}{k}$.
        Assume there is an algorithm for $(k,2^m, r)$-XOR, with success
        probability $\beta$ and expected running time $\mathcal{T}$.
		Then, there is an algorithm for $(k, 2^n, r)$-XOR with success probability
		\[
			\frac{\beta^4}{128} \cdot \Big(2^n/\binom{r}{k}\Big)^2
		\]
		and expected running time at most $2^{n-m} \cdot (\mathcal{T}+ \tilde{O}(r))$.
\end{lemma}

\begin{lemma}[Simple reduction]
\label{lem:sparse}
    Let $d > 0$ by an integer. Assume there is an algorithm for $(k,2^n, d \cdot r)$-XOR, with success
    probability $\beta$ and expected running time $\mathcal{T}$.
	Then, there is an algorithm for $(k, 2^n, r)$-XOR with success probability at least $\tfrac{\beta}{(2d)^k}$ and expected
    running time at most $\mathcal{T} + \tilde{O}_k(d \cdot r)$.
\end{lemma}
\begin{proof}[Proof of Lemma~\ref{lem:sparse}]
Let $B$ be an algorithm for $(k, 2^n, d \cdot r)$-XOR.
We construct an algorithm for $(k, 2^n, r)$-XOR as follows. Given an instance $\z_1,\ldots,\z_r$,
pick $(d-1) \cdot r$ vectors in $\{0,1\}^n$
uniformly at random and append them the original instance. Then, apply a uniform permutation to the
$d \cdot r$ vectors and run $B$ on this instance.
If $B$ succeeds to return a $k$-XOR, and this $k$-tuple is included in $\z_1,\ldots,\z_r$, then the algorithm returns it.
Otherwise, it fails.

As the input to $B$ is uniformly permuted,
the events (1) $B$ succeeds, and (2) the returned $k$-tuple is included in $\z_1,\ldots,\z_r$,
are independent. By assumption, the probability of the first event is $\beta$. Note that we may assume $r \geq 2k$, as otherwise, a trivial algorithm for $(k,2^n,r)$-XOR that goes over all possible $k$-tuples of input vectors runs in time $O_k(1) \leq \mathcal{T}+\tilde{O}_k(d \cdot r)$.
Given that $r \geq 2k$, the probability of the second event is at least
$((r-k)/(d \cdot r - k))^k \geq (2d)^{-k}$.
Thus, the algorithm succeeds with probability at least $\beta/(2d)^k$.
\end{proof}

We now prove that Theorem~\ref{thm:red} follows from the lemmas.
\begin{proof}[Proof of Theorem~\ref{thm:red}]
Our goal is to devise an algorithm for $(k, 2^n, 2^{n/k})$-XOR
given an algorithm $B$ for $(k, 2^m, 2^{n/k})$-XOR with
success probability $\beta$ and expected running time $\mathcal{T}$.

We first construct an algorithm for $(k, 2^n, 2^{n/k} \cdot k)$-XOR
and then use Lemma~\ref{lem:sparse} to sparsify the input.

Clearly, $B$ is also applicable to $(k, 2^m, 2^{n/k} \cdot k)$-XOR
with similar success probability and complexity
(by ignoring all but the first $2^{n/k}$ input vectors).
Denote $r = 2^{n/k} \cdot k$ and note that as $k! > (k/e)^k$ for all $k$, we have
\[
	e^{-k}\binom{r}{k} \leq 2^n = (r/k)^k \leq \binom{r}{k}.
\]
Therefore, based on algorithm $B$ for $(k, 2^m, 2^{n/k} \cdot k)$-XOR,
by Lemma~\ref{lem:mainxor}, there is an algorithm $B_1$
for $(k, 2^n, 2^{n/k} \cdot k)$-XOR
with success probability
\[
	\beta_1 = \frac{\beta^4}{128} \cdot \Big(\frac{2^n}{\binom{r}{k}} \Big)^2 \geq \frac{\beta^4}{128 e^{2k}}
\]
and expected
running time at most
$2^{n-m} \cdot (\mathcal{T}+ \tilde{O}_k(2^{n/k}))$.

Using algorithm $B_1$ for $(k, 2^n, r = 2^{n/k} \cdot k)$-XOR, Lemma~\ref{lem:sparse}
(applied with $d = k$) implies that there is
an algorithm for $(k, 2^n, 2^{n/k})$-XOR
with success probability at least
\[
	\frac{\beta_1}{(2d)^k} \geq \frac{\beta^4}{(16k)^{k+2}}
\]
and expected running time at most
$2^{n-m} \cdot (\mathcal{T}+ \tilde{O}_k(2^{n/k})) + \tilde{O}_k(k \cdot 2^{n/k}) = 2^{n-m} \cdot (\mathcal{T}+ \tilde{O}_k(2^{n/k}))$,
as claimed.
\end{proof}

\subsection{Overview of the main reduction lemma}

The proof of the main reduction lemma (Lemma~\ref{lem:mainxor}) is constructive -- namely, we construct an algorithm and show that it satisfies the assertion of the lemma.  A natural way to solve the sparse $k$-XOR problem using an oracle for the dense $k$-XOR problem, is to \textbf{truncate} some of the bits. That is, given $r$ vectors, $\z_1, \ldots, \z_r \in \zo^n$, with $2^n \approx r^k$ (so that we expect only $\Theta(1)$ solutions) we may feed the oracle with $x_1, \ldots, x_r \in \zo^m$, where $x_i$ is obtained from $\z_i$ by truncating the last $t\defeq n-m$ bits. In the new problem corresponding to $x_1, \ldots, x_r$ we expect to have $\Theta(2^t)$ (i.e., many) solutions, and hence the oracle is applicable. A $k$-tuple output by the dense $k$-XOR oracle has a probability of $2^{-t}$ to be a solution to the original $\z$-problem (since we truncated exactly $t$ bits, and $\z$ is uniformly distributed).
Thus, it seems that if we feed the oracle $\Theta(2^t)$ times with truncated inputs, we expect $\Theta(1)$ out of the $\Theta(2^t)$ output $k$-tuples to solve the original $\z$-problem, and consequently, solving the $x$-problem is at most $2^t$ times easier than solving the $\z$-problem.

The \textbf{flaw} in this argument is that we cannot expect the oracle to output a newly-forged $k$-tuple in every application (especially if the oracle is deterministic and is fed with the same inputs in all applications).
Hence, although the expected number of solutions we find is $\Theta(1)$, it might be that we solve the $\z$-problem only with a small probability (e.g., it might be that with a high probability the oracle outputs many identical solutions).

Therefore, we have to \textbf{trick} the oracle so that the $k$-tuples it outputs in the different applications will be \textbf{pairwise independent} of each other -- almost as if we feed it with a fresh uniformly chosen input every time.

For this purpose, we devise a method that receives $r$ vectors in $\zo^n$ (denoted $\z_1,\ldots,\z_r$) and randomly obfuscates them, returning $r$ vectors in $\zo^m$ (denoted $x_1,\ldots,x_r$).
This obfuscation meets two criteria:
\begin{enumerate}[noitemsep]
 	\item A solution to the $k$-XOR $x$-problem gives rise to a solution to the $k$-XOR $\z$-problem with good probability, (i.e., $p \approx 2^{-t}$).

 	\item The obfuscation is powerful enough to disguise the fact all the $x$'s are generated from the same $\z$, so that each application of the oracle \textbf{(pairwise) independently} has a chance to solve the $\z$-problem.
\end{enumerate}
As we show below, this obfuscation method guarantees that after applying the oracle sufficiently many times on obfuscated $x$-problems, with a high probability a solution of the original $\z$-problem will be obtained.

\begin{remark}
	The proof crucially relies on the strength of the obfuscation -- specifically, on the fact that the probability of outputting the same $k$-tuple in an iteration pair is $O(2^{-2t})$. This is the motivation behind using the obfuscation we propose, as weaker obfuscations, such as \textbf{truncating $t$ randomly chosen bits}, have oracles that output the same $k$-tuple in an iteration pair with probability much higher than $2^{-2t}$, even if we apply an invertible linear transformation after truncation. 

In detail, consider a $k$-tuple of vectors whose XOR is $v \in \zo^n$ ($v \neq 0$). Then, the probability that after randomly truncating bits from this $k$-tuple, they would XOR to $0$ is higher for $v$ of low Hamming weight. Now, suppose that there is a vector in the $x$-problem that belongs to several $k$-tuples whose XOR has low (non-zero) Hamming weight. Then, the corresponding vector in the $z$-problem is expected to belong to more $k$-XORs than the average vector. This vector can thus be singled out by the oracle, which would repeatedly output one of the $k$-XORs it belongs to (with relatively high probability). 
\end{remark}

\paragraph{Structure of the proof.} First we present the obfuscation algorithm and the lemma which asserts that it indeed satisfies the aforementioned properties. Then we prove the main reduction lemma, assuming the obfuscation lemma. Finally, we prove the obfuscation lemma, which is the most complex part of the conditional $k$-XOR hardness proof.

\subsection{The obfuscation algorithm}\label{ssec:XOR-obfuscation}

Let $m,n,r$ satisfy $n/2 \leq m \leq n \leq \log_2 \binom{r}{k}$, and let $L$ be a parameter to be specified below. Let $B$ be an algorithm for $(k,2^m, r)$-XOR. The algorithm $A$ for $(k,2^n, r)$-XOR, which receives as input an $r$-tuple of $n$-bit vectors $(\z_1,\ldots,\z_n) \in \{0,1\}^{r \times n}$, is defined as follows.

\begin{alg}\label{alg:xor}\skipline
		\begin{enumerate}[noitemsep]
			\item \textbf{Repeat} $L$ times:
			\item
			\begin{adjustwidth}{5mm}{}
			\textbf{Draw} a uniformly random full-rank matrix $T \in \field_2^{m \times n}$ (rank $= m$) and a uniformly random permutation $P \in S_r$.
			\end{adjustwidth}
			\item
			\begin{adjustwidth}{5mm}{}
			\textbf{Let} $x_{i} = T(\z_{P(i)})$ for all $i \in [r]$.
			\end{adjustwidth}
			\item
			\begin{adjustwidth}{5mm}{}
			\textbf{Feed} $B$ with $(x_1, \ldots, x_r)$. In case it outputs a $k$-tuple $K$ with $\bigoplus_{i \in K} x_i = 0_m$, \textbf{test} whether $\mathcal{K} = P(K)$ satisfies $\bigoplus_{i \in \mathcal{K}} \z_{i} = 0_n$, and \textbf{if} it does -- \textbf{output} the $k$-tuple $\mathcal{K}$. \textbf{Otherwise, continue}.
			\end{adjustwidth}
		\end{enumerate}
\end{alg}

		Thus, we try to solve $A$'s problem, by considering many derived problems (in which we consider only $m$-bit vectors, instead of $n$-bit ones), trying to solve these using $B$, and in case its output solves $A$'s problem, we output the result. Each of these trials succeeds with some probability, and only one success is required. Repeating this procedure enough times, we are expected to find a solution with reasonable probability -- unless failures are correlated. This is where the obfuscation lemma comes into play -- it shows the trials are sufficiently independent for $A$ to succeed.
		
\paragraph{The obfuscation lemma.} The heart of the proof is the following lemma:
\begin{lemma}\label{Cor:Obfuscation}
	Let $(\z_1,\ldots,\z_r) \in \{0,1\}^{r \times n}$ be chosen uniformly at random. Let $(x_1^{(1)},\ldots,x_r^{(1)}) \in \{0,1\}^{r \times m}$ and $(x_1^{(2)},\ldots,x_r^{(2)}) \in \{0,1\}^{r \times m}$ be obtained from it by the procedure described above (in two out of the $L$ iterations). Let $\mathcal{K}_1, \mathcal{K}_2$ be the two corresponding $\mathcal{K}$'s obtained in the process. Then, assuming $t \defeq n-m \leq m \leq n \leq \log_2 \binom{r}{k}$, we have
	\begin{equation}\label{eq:indep}
		\Pr[\mathcal{K}_1 = \mathcal{K}_2] \leq 2^{2-2t},
	\end{equation}
	where the probability is taken over $z,x^{(1)}, x^{(2)}$, and $B$'s randomness.
\end{lemma}

Note that cases where at least one of $\mathcal{K}_1,\mathcal{K}_2$ is not obtained (that is, when in at least one of the two iterations, Algorithm~$B$ fails to find a solution to the $(k,2^m, r)$-XOR problem) are not counted as equality between $\mathcal{K}_1$ and $\mathcal{K}_2$.

We note that each of $(x_1^{(1)},\ldots,x_r^{(1)})$ and $(x_1^{(2)},\ldots,x_r^{(2)})$ is likely to have about $2^{t}$ solutions to $k$-XOR, but only $O(1)$ of these are common to both and typically correspond to $k$-XOR solutions for $(\z_1,\ldots,\z_r)$. Therefore, the lemma essentially asserts that $B$ cannot do much better than output a uniform solution to each $x$-problem.

\subsection{Proof of the main reduction lemma}
\label{sec:mainred}

We prove now that the assertion of Lemma~\ref{lem:mainxor} follows from Lemma~\ref{Cor:Obfuscation}. The proof is a rather standard probabilistic argument. Afterwards we present the considerably more complex proof of Lemma~\ref{Cor:Obfuscation}.

\begin{proof}[Proof of Lemma~\ref{lem:mainxor}, assuming Lemma~\ref{Cor:Obfuscation}]
		Consider a slightly tweaked obfuscation process which has exactly $L$ iterations (and may output multiple solutions). Clearly, the success probability of the tweaked obfuscation process is identical to the original one, and thus we analyze it instead.
		For any $1 \leq l \leq L$, let $\mathcal{K}_l$ be the $k$-tuple obtained in the $l$'s iteration ($\mathcal{K}_l$ exists only when B succeeds, i.e. with probability $\beta$).
		Denote by $S_l$ the event that $\mathcal{K}_l$ admits a solution to the $(k, 2^n, r)$-XOR problem.
		We have $\Pr[S_l]= 2^{m-n}\beta$ for each $l = 1,\ldots,L$, since $z$ is uniformly random, and $m$ out of the $n$ dimensions of $\bigoplus_{i \in \mathcal{K}_l} z_i$ are known to nullify, independently of the other, beforehand-erased, $n - m$ dimensions. (In other words, we know that $\bigoplus_{i \in \mathcal{K}_l} z_i$ belongs to the kernel of a randomly chosen full-rank linear transformation $T:\field_2^{n} \to \field_2^m$, and hence, $\Pr[\bigoplus_{i \in \mathcal{K}_l} z_i=0_n]=2^{m-n}$.)

		Define the random variables
		\[
		Z'
		\defeq
		\sum_{l=1}^{L} \one\{S_l\} -\sum_{1 \leq l < l' \leq L} \one\{\mathcal{K}_l = \mathcal{K}_{l'}\} , \qquad Z \defeq \max(Z', 0),
		\]
		where $\one\{E\}$ is the indicator of the event $E$.
		A simple inclusion-exclusion-like principle shows that $Z'$ lower bounds the number of distinct solutions found for the $(k, 2^n, r)$-XOR problem in (tweaked) Algorithm~\ref{alg:xor}. As the number of solutions is non-negative, $Z$ lower bounds it as well.
		The Paley-Zygmund inequality~\cite[ch.2]{BoucheronL13}, applied for the non-negative random variable $Z$, implies
		\[
			\pr[Z > 0] \geq \frac{\be[Z]^2}{\be[Z^2]}.
		\]
		Since $\pr[Z > 0]$ lower bounds the probability that (tweaked) Algorithm~\ref{alg:xor} solves the $(k, 2^n, r)$-XOR problem, our task is reduced to lower bounding $\be[Z]^2$, and upper bounding $\be[Z^2]$. The value of $\be[Z]$ is easily bounded as
		\[
			\be[Z] \geq \be[Z'] = \sum_{l=1}^{L} \pr[S_l] - \sum_{1 \leq l < l' \leq L} \pr[\mathcal{K}_l = \mathcal{K}_{l'}].
		\]
		Using $\be[S_l] = 2^{m-n} \beta$ and Lemma~\ref{Cor:Obfuscation}, we obtain $\be[Z] \geq L\cdot 2^{m-n}\beta - \binom{L}{2} 2^{2(1+m-n)}$. We choose $L=\beta \cdot 2^{n-m-2}$, and obtain
		\[
		\be[Z]\geq \beta^2/8.
		\]
		We henceforth upper bound $\be[Z^2]$.
		Let $D$ be the number of distinct $k$-tuples $T\seq [r]$ with $\bigoplus_{i \in T} z_i = 0_n$ (that is, the number of actual solutions for the $(k, 2^n, r)$-XOR problem in the set $\{\z_1,\ldots,\z_r\}$).
		Since $Z$ is not larger than the number of solutions found in (tweaked) Algorithm~\ref{alg:xor}, we have $Z \leq D$, and in particular, $\be[Z^2] \leq \be[D^2]$.
		
		Note that for different $k$-tuples $T, T'$, the events $\bigoplus_{i \in T} z_i = 0_n$ and $\bigoplus_{i \in T'} z_i = 0_n$ are independent, and each of them has probability $2^{-n}$. Hence,
		\[
			\be[D^2] = \binom{r}{k}2^{-n} + \binom{r}{k} \left(\binom{r}{k} - 1 \right) 2^{-2n} \leq \binom{r}{k}^2 2^{1-2n},
		\]
		where the ultimate inequality holds since $\binom{r}{k} \geq 2^n$ by assumption. Therefore, the algorithm succeeds with probability at least
		\[
			\pr[Z > 0] \geq \frac{\be[Z]^2}{\be[D^2]} \geq 			
			\frac{\beta^4}{128 \binom{r}{k}^2 2^{-2n}}.
		\]
		The running time of Algorithm~\ref{alg:xor} (including the $\wt{O}(r)$ additional overhead of each iteration)
		is
		\[
			L \cdot (\mathcal{T} + \wt{O}(r)) \leq 2^{n-m} \cdot (\mathcal{T}+\wt{O}(r)).
		\]
		This completes the proof of the lemma.
\end{proof}


\subsection{Proof of the obfuscation lemma}

In this section we prove Lemma~\ref{Cor:Obfuscation}. We start by introducing a distribution that models two independent outputs of the obfuscation process, and restate the obfuscation lemma.

	\begin{definition}\label{def:noise}
	We say that a pair of random variables $(x^{(1)},x^{(2)})$, each taking values in $\field_2^{r \times m}$, has an $(m,r,t)$-distribution, if there exist random variables
	$\z,\, T^{(j)}$ for $j \in \{1,2\}$
	where:
	\begin{enumerate}[noitemsep]
		\item $\z,\, T^{(1)},\, T^{(2)}$,
		are independent random variables.

		\item $\z \sim \field_2^{r \times (m+t)}$ is uniformly distributed.

		\item $\func{T^{(j)}}{\field_2^{m+t}}{\field_2^m}$ is a uniformly random, full-rank (i.e., rank $=m$), linear transformation.

		\item $x^{(j)}_i = T^{(j)}(\z_i)$.
	\end{enumerate}
	\end{definition}

	\begin{lemma}\label{lem:xs-bound}
		Let $B$ be an algorithm that receives as input a list of $r$ vectors, each of length $m$ bits, and outputs the indices of $k>0$ vectors among them whose XOR is $0_m$ (or a failure string).
		If $(x,y)$ has an $(m,r,t)$-distribution, and $P,Q\sim S_r$ are two uniformly random and independent permutations, then
		\begin{equation}\label{eq:xs-bound}
			\Pr[P^{-1}(B(P(x))) = Q^{-1}(B(Q(y)))] \leq 2^{-2t} +  \frac{2^{m-t}}{\binom{r}{k}} + 2^{-t+1-m},
		\end{equation}
		where the probability is taken over $B$'s randomness, $x,y$ and $P,Q$ (The event on the left hand
		side is contained in the event that both executions $B(P(x))$, $B(Q(y))$ succeed).
	\end{lemma}
	
	Notice that Lemma~\ref{Cor:Obfuscation} immediately follows from Lemma~\ref{lem:xs-bound} (compare~\eqref{eq:indep} with~\eqref{eq:xs-bound}).

\subsubsection{Proof outline} The proof of Lemma~\ref{lem:xs-bound} uses techniques from \emph{discrete Fourier analysis} and consists of several steps.
\begin{enumerate}
	\item \textbf{Transformation to real-valued functions.}
	We show that instead of analyzing the obfuscation on a tuple-valued function, it is sufficient to analyze its action on the simpler class of real-valued functions.
	We utilize the fact that our obfuscation randomly permutes the input vectors, so that any oracle $\func{B}{\zo^{r\times m}}{\binom{[r]}{k}}$ must, informally, treat all candidate output $k$-tuples in the same way.
	Hence it suffices to analyze the modified, real-valued, oracle $\func{B'}{\zo^{r\times m}}{[0,1]}$ which indicates the probability that $B$ outputs the specific $k$-tuple $K \defeq \{1,\ldots,k\}$ when applied on its input. Specifically, our task is reduced to showing that
	\begin{equation}\label{eq:intro-to-bounded}
	 	\be[B'(y) B'(y')] \leq O_k(2^{-2t} / r^k),
	 \end{equation}
	where $y, y'$ are two independent obfuscations of a common, random, $\z \in \zo^{r\times n}$.
	
	\item \textbf{Bounding the correlation using discrete Fourier analysis.} In order to prove~\eqref{eq:intro-to-bounded}, we consider the Fourier expansion of $B'$, namely,
	\[B'=\sum \wh{B}'(S) \chi_S, \qquad \mbox{where} \qquad \chi_S(v) = (-1)^{\sum_{i \in S}v_i}.
	\]
	We divide the Fourier expansion into two parts -- the \emph{Cartesian} part
	\[(B')^C = \sum_{\{S = U \times V: U \subseteq [r], V \subseteq [m]\}} \wh{B'}(S)\chi_S,
	\]
	and the non-Cartesian part $(B')^{\perp}=B'-(B')^C$ (which is orthogonal to $(B')^C$). Informally, the contribution of the Cartesian part to the correlation corresponds to the information on \emph{aligned XORs} of variables (such as $(\z_{1,2} \oplus \z_{1,3}) \oplus (\z_{4,2} \oplus \z_{4,3})$) preserved between the two obfuscations, while the contribution of the non-Cartesian part carries the rest of the information. Then we handle each part of the correlation separately:
		\begin{enumerate}
			\item \emph{Obfuscation hides everything but aligned XORs.} We show that for any function $B'$, the obfuscation reduces the contribution to the left hand side of~\eqref{eq:intro-to-bounded}, associated with the non-Cartesian part, to at most $2^{-2t}$. This argument depends only on the obfuscation, and does not rely on the specific problem we try to solve.
			
			\item \emph{Aligned XORs do not reveal much.} We show that in the case of the $k$-XOR problem, the contribution of the Cartesian part is also small. Here we use the specific structure of the problem -- specifically, the set $\{x:B'(x) >0\}$ being small and admitting a nice algebraic structure (namely, $B'(x)=0$ whenever $\bigoplus_{i=1}^k x_i \neq 0_m$).
		\end{enumerate}
	Combination of the two bounds completes the proof.
\end{enumerate}

\subsubsection{Transformation to real-valued functions}

	\begin{lemma}\label{lem:to-bounded}
		Let $\func{B}{\zo^{r\times m}}{\binom{[r]}{k}}$ be an algorithm that outputs either a $k$-tuple $R$ with $\bigoplus_{i \in R} x_i = 0_m$, or a failure string.
		Let $K \defeq \{1,\ldots , k\}$ and define $\func{B'}{\zo^{r\times m}}{[0,1]}$ by
		\begin{equation}\label{eq:2}
			B'(x) = \be_{\substack{P \sim S_r \\\bar{T} \sim GL_m(\field_2)}}\li[ \one\{B(\bar{T}(P(x))) = P(K)\}\ri],
		\end{equation}
		where $P \sim S_r$ is a uniformly random permutation, and $\bar{T} \sim GL_m(\field_2)$ is a uniformly random invertible linear map. The expectation in~\eqref{eq:2} is taken also over $B$'s randomness.
		Then,
		\begin{equation}\label{eq:annihilate}
			\bigoplus_{i \in K} x_i \neq 0 \  \implies \  B'(x) = 0,
		\end{equation}
		\begin{equation}\label{eq:measure-bound}
			\be_x[B'(x)] \leq 1/\binom{r}{k},
		\end{equation}
		and if $(x,y)$ has an $(m,r,t)$-distribution and $P,Q \sim S_r$ are independent, then
		\begin{equation}\label{eq:to-bounded-need}
			 \Pr[P^{-1}(B(P(x))) = Q^{-1}(B(Q(y)))] = \binom{r}{k} \be_{x,y}[B'(x) B'(y)].
		\end{equation}
	\end{lemma}

\begin{remark}
We note that while the obfuscation algorithm uses a full-rank shrinking transformation $T$ from $\field_2^{m+t}$ to $\field_2^m$, this transformation does not appear explicitly in Lemma~\ref{lem:to-bounded}. Instead, it appears implicitly via the assumption that $(x,y)$ has an $(m,r,t)$-distribution (made just before~\eqref{eq:to-bounded-need}), and plays a central role in the proof of~\eqref{eq:to-bounded-need}.
\end{remark}

		\begin{proof}
			To show~\eqref{eq:annihilate} note that if $\bigoplus_{i \in K} x_i \neq 0$ then $B$ cannot output $P(K)$ on the input $\bar{T}(P(x))$, by our assumption on $B$, and $\bar{T}$ being invertible. Hence $B'(x) = 0$ in such a case.

			To verify~\eqref{eq:measure-bound}, denote $x' = \bar{T}(P(x))$ and observe that when $x \sim \{0,1\}^{r \times m}$, we have  $x' \sim \{0,1\}^{r \times m}$ independently of $P$. Hence, by interchanging order of summation,
			\begin{align*}
				\be_{x}[B'(x)] &= \be_{P,\bar{T}}[\be_{x}[\one\{B(\bar{T}(P(x))) = P(K)\}]] = \be_{P,\bar{T}}[\be_{x'}[\one\{B(x') = P(K)\}]] \\
				&=  \be_{x'}[\be_{P}[\one\{B(x') = P(K)\}]] \leq 1/\binom{r}{k},
			\end{align*}
			where the latter inequality holds because for any fixed $x'$, $P(K)$ attains the value of $B(x')$ with probability at most $1/\binom{r}{k}$.

			In order to prove~\eqref{eq:to-bounded-need}, we reason about $\be_{x,y}[B'(x) B'(y)]$. Observe that for any $K' \seq [r]$ with $|K'| = k$, the function $B'_{K'}$ defined by $B'_{K'}(x) = \be_{P,\bar{T}}[\one\{B(\bar{T}(P(x))) = P(K')\}]$ satisfies
			\begin{equation}\label{eq:3}
			\be_{x,y}[B'_{K'}(x) B'_{K'}(y)] =  \be_{x,y}[B'(x) B'(y)].
			\end{equation}
			Indeed, let $R \in S_r$ be such that $R(K) = K'$. As $(R(x), R(y))$ has the same distribution as $(x, y)$, we have
			\[
			\begin{aligned}
				\mathbb{E}_{x,y}[B'_{K'}&(x) B'_{K'}(y)]
				= \mathbb{E}_{x,y}[B'_{K'}(R(x)) B'_{K'}(R(y))] \\
				&=
				\be_{x,y}\big[\be_{P', \bar{T}'}[\one\{B(\bar{T}'(P'R(x)))=P'(K')\}] \be_{P'', \bar{T}''}[\one\{B(\bar{T}''(P''R(y)))=P''(K')\}]\big] \\
				&= \be_{x,y}\big[\be_{P', \bar{T}'}[\one\{B(\bar{T}'(P'(x)))=P'R^{-1}(K')\}] \be_{P'', \bar{T}''}[\one\{B(\bar{T}''(P''(y)))=P''R^{-1}(K')\}]\big] \\
				&=
				\be_{x,y}[B'(x) B'(y)].
			\end{aligned}
			\]
			Notice that if $(x,y)$ has an $(m,r,t)$-distribution, and $\bar{T}',\bar{T}'' \sim GL_m(\field_2)$ are uniformly random invertible linear maps independent of all other variables, then $(\bar{T}'(x), \bar{T}''(y))$ has an $(m,r,t)$-distribution as well. We verify~\eqref{eq:to-bounded-need}:
			\[
			\begin{aligned}
				\Pr_{x,y,P,Q}& \big[P^{-1}(B(P(x))) = 
				Q^{-1}(B(Q(y)))\big] \\
				&= \sum_{K'}  \be_{x,y,P,Q}\big[\one\{P^{-1}(B(P(x)))=K'\} \one\{Q^{-1}(B(Q(y)))=K'\} \big] \\
				&= \sum_{K'} \be_{x,y,P,Q,\bar{T}', \bar{T}''}\big[\one\{P^{-1}(B(P(\bar{T}'(x))))=K'\} \one\{Q^{-1}(B(Q(\bar{T}''(y))))=K'\} \big] \\
				&= \sum_{K'} \be_{x,y}\big[B'_{K'}(x) B'_{K'}(y)\big] \\
				&= \binom{r}{k} \be_{x,y} \big[ B'(x) B'(y)  \big],
			\end{aligned}
			\]
			where the penultimate equality holds since $P$ (resp.~$Q$) commutes with $\bar{T}'$ (resp.~$\bar{T}''$), and the ultimate equality uses~\eqref{eq:3}.
		\end{proof}

	\subsubsection{Obfuscation hides everything but aligned XORs}\label{ssec:counts}

We begin with the standard definition of the Fourier-Walsh expansion of functions on the discrete cube, adapted to our setting.
		\begin{definition}[Fourier expansion]
			Given $S \seq [r] \times [m]$, define $\func{\chi_{S}}{\zo^{r\times m}}{\spm}$ by $\chi_S(x) = (-1)^{\sum_{(i, j) \in S} x_{i,j}}$. The set $\{\chi_S\}_{S \seq [r]\times [m]}$ is an orthonormal basis for the set of functions $\set{f}{\func{f}{\zo^{r\times m}}{\reals}}$, with respect to the standard inner product $\langle f,g \rangle = \be_{x\sim \zo^{r\times m}}[f(x)g(x)]$. Hence each $\func{f}{\zo^{r\times m}}{\reals}$ can be decomposed to
			\[
				f = \sum_{S \seq [r] \times [m]} \wh{f}(S) \chi_S,
			\]
			where $\wh{f}(S) = \langle f,\chi_S \rangle$, and in particular, $\wh{f}(\es) = \mathbb{E}[f]$.
		\end{definition}

		\begin{definition}[Cartesian decomposition]
			Given $S \seq [r] \times [m]$, we call $S$ a \emph{Cartesian product} if there exist $U \seq [r]$ and $V \seq [m]$ such that $S = U \times V$.

			The Fourier expansion allows decomposing any function $\func{f}{\zo^{r\times m}}{\reals}$ into Cartesian and non-Cartesian parts:
			\[
				f = f^{C} + f^{\perp} = \li( \sum_{S\text{ Cartesian product}} \wh{f}(S) \chi_S \ri) + \li( \sum_{S\text{ non Cartesian product}} \wh{f}(S) \chi_S\ri),
			\]
			where $\langle f^{C}, f^{\perp} \rangle = 0$.
		\end{definition}

		\begin{definition}[Cartesian functions]
			A function $\func{f}{\zo^{r\times m}}{\reals}$ is called a \emph{Cartesian function} if $f=f^C$.
		\end{definition}

		\begin{lemma}\label{lem:noise}
			Let $\func{f}{\zo^{r\times m}}{\reals}$. Suppose $(x, y)$ has an $(m,r,t)$-distribution. Then
			\[
				\cov(f(x), f(y)) \leq 2^{-t} \norm{f^{C}}_2^2 + 2^{-2t} \norm{f^{\perp}}_2^2.
			\]
		\end{lemma}
		
		Recall that the goal of the obfuscation process is to reduce the correlation between different obfuscations of the same input $z$ (which correspond to different iterations of Algorithm~\ref{alg:xor} described above) to $O_k(1) \cdot 2^{-2t}$. In this respect, the lemma asserts that the obfuscation hides the non-Cartesian part, which corresponds to everything except for \emph{aligned XORs} (that is, expressions of the form $\bigoplus_{i \in I} \bigoplus_{j \in J} x_{i,j}$).
		
		\begin{proof}
			Write $f = \sum_{S} \wh{f}(S) \chi_S$. We have
			\begin{equation}\label{eq:noise-compute1}
			\begin{aligned}
				\cov(f(x), f(y))
				&=
				\be_{(x,y)}\big[(\sum_{S} \wh{f}(S) \chi_S(x))(\sum_{S'} \wh{f}(S') \chi_{S'}(y))\big] - \be[f]^2
				\\ & =
				\be\big[\sum_{S, S'} \wh{f}(S) \wh{f}(S') \chi_S(x) \chi_{S'}(y) \big] - \be[f]^2
				\\ & \leq
				\sum_{S} \wh{f}(S)^2 \Big(\sum_{S'} \big|\be\big[ \chi_S(x) \chi_{S'}(y) \big] \big| \Big)   - \wh{f}(\es)^2,
			\end{aligned}
			\end{equation}
			where the last step uses the inequality $\wh{f}(S) \wh{f}(S') \leq (\wh{f}(S)^2+\wh{f}(S')^2)/2$, applied for all $S,S'$.
			
			In order to analyze $\be\big[ \chi_S(x) \chi_{S'}(y) \big]$, let us recall how $(x,y)$ is distributed according to Definition~\ref{def:noise}. We draw a uniformly random $\z \sim \field_2^{r\times (m+t)}$ and two uniformly random rank-$m$ linear maps $\func{T_1, T_2}{\field_2^{m+t}}{\field_2^m}$ and define $(x, y) = (T_1(\z), T_2(\z))$.
		
Observe that there exist linear maps $T^{\ast}_1, T^{\ast}_2: (\field_2^{m})^{\ast} \to (\field_2^{m+t})^{\ast}$ such that for each $S=(S_1,\ldots,S_r) \subseteq [r] \times [m]$, we have
\[
\chi_S(x)=\chi_S(T_1(\z))=\chi_{T^{\ast}_1 S}(\z) \qquad \mbox{and} \qquad \chi_S(y)=\chi_S(T_2(\z))=\chi_{T^{\ast}_2 S}(\z),
\]
where the $S_i$'s are regarded as elements of $(\field_2^{m})^{\ast} \cong \field_2^{m}$.

Formally, consider the dual linear maps $T^{\ast}_1, T^{\ast}_2: (\field_2^{m})^{\ast} \to (\field_2^{m+t})^{\ast}$, defined by
			\[
				T^{\ast}_l(f)(a) \defeq f(T_l(a)),\qquad l=1,2
			\]
			where $\func{f}{\field_2^m}{\field_2}$ is a linear functional, and $a \in \field_2^{m+t}$ is a vector. (Matrix-wise, the representing matrix of $T^{\ast}_l$ according to the (dual) standard basis, is the transpose of the representing matrix of $T_l$ in the standard basis.) Note that each $S \subseteq [r] \times [m]$ naturally corresponds to an $r$-tuple of linear functionals $(S_1, \ldots, S_r)$, where $S_i(b)=\bigoplus_{j=1}^{m} ( b_j \cdot \one\{(i,j) \in S\} )$ for any $b \in \field_2^m$. Thus, we may slightly abuse notation and write $S=(S_1,\ldots,S_r)$, and subsequently, define $T^{\ast}_l(S) = (T^{\ast}_l(S_1), \ldots, T^{\ast}_l(S_r))$ and regard its outputs as elements of $[r] \times [m+t]$.

			Since for each $S=(S_1,\ldots,S_r)$, we have $\chi_S(x)=\chi_S(T_1(x))=\chi_{T^{\ast}_1 S}(\z)$ and
			$\chi_{S'}(y)=\chi_{T^{\ast}_2 S'}(\z)$ as was written above, and since
			$\be[\chi_{A}(\z) \chi_{B}(\z)] = \one\{A=B\}$, we can write~\eqref{eq:noise-compute1} as
			\begin{align*}
				\cov(f(x), f(y))
				&\leq
				\sum_{S} \wh{f}(S)^2 \Big(\sum_{S'} \big|\be_{T_1,T_2}\be_{\z} \big[ \chi_{T^{\ast}_1 S}(\z) \chi_{T^{\ast}_2 S'}(\z) \big] \big| \Big)   - \wh{f}(\es)^2 \\
				&\leq
				\sum_{S} \wh{f}(S)^2 \Big(\sum_{S'} \Pr\big[T_1^{\ast}(S) = T_2^{\ast}(S')\big] \Big)   - \wh{f}(\es)^2.
			\end{align*}
			Noting that $T_2^{\ast}$ is injective (as $T_2$ is of full rank, and duality preserves rank), we conclude
			\[
				\cov(f(x), f(y))
				\leq
				\sum_{S \neq \es} \wh{f}(S)^2 \Pr\big[T_1^{\ast}(S) \in \im(T_2^{\ast})\big] =
				\sum_{S \neq \es} \wh{f}(S)^2 \Pr\Big[ \bigwedge_{i=1}^r \big[T_1^{\ast}(S_i) \in \im(T_2^{\ast})\big]\Big].
			\]		

			To see that $\Pr[T_1^{\ast}(S) \in \im(T_2^{\ast})] \leq 2^{-t}$ for all $S \neq \es$, choose a nonempty row $S_i$ in $S$, and observe that the probability $\Pr[T_1^{\ast}(S_i) \in \im(T_2^{\ast})]$ is the same as the probability that a specific non-zero vector in $\field_2^{m+t}$ is inside a random subspace of $\field_2^{m+t}$ of dimension $m$. This probability is $\frac{2^m-1}{2^{m+t}-1} \leq 2^{-t}$.

			We furthermore claim that $\Pr[T_1^{\ast}(S) \in \im(T_2^{\ast})] \leq 2^{-2t}$ whenever $S$ is not a Cartesian product. Indeed, if we choose two different nonempty `rows' $S_i$, $S_j$ of $S$ (which is possible as $S$ is not a Cartesian product), the probability $\Pr[T_1^{\ast}(S) \in \im(T_2^{\ast})]$ is upper bounded by the probability that both $T_1^{\ast}(S_i) \in \im(T_2^{\ast})$ and $T_1^{\ast}(S_j) \in \im(T_2^{\ast})$. Since $T_1^{\ast}(S_i)$, $T_1^{\ast}(S_j)$ are distinct, the aforementioned probability is $\frac{2^m-1}{2^{m+t}-1} \cdot \frac{2^m-2}{2^{m+t}-2} \leq 2^{-2t}$.
			Overall,
			\[
			\begin{aligned}
				\cov(f(x), f(y))
				&\leq
				2^{-t} \sum_{S\text{ Cartesian}} \wh{f}(S)^2 + 2^{-2t} \sum_{S\text{ not Cartesian}} \wh{f}(S)^2 \\
				& = 2^{-t} \norm{f^{C}}_2^2 + 2^{-2t} \norm{f^{\perp}}_2^2,
			\end{aligned}
			\]
			where the ultimate equality uses Parseval's identity. This completes the proof.
		\end{proof}

		\subsubsection{Aligned XORs do not reveal much}
		
		\begin{lemma}\label{lem:non-cartesian}
		Let $\func{B'}{\zo^{r\times m}}{\czo}$ have $\be_x[B'(x)]=\mu$. Suppose $B'(x)=0$ whenever $\bigoplus_{i=1}^{k} x_i \neq 0_m$ ($k > 0$). Then
		\[
			\norm{(B')^C}_2^2 \leq 2^m \mu^2 + 2^{1-m} \mu.
		\]
		\end{lemma}
		
		\begin{proof}
		Write $f(x) = \one\{\bigoplus_{i=1}^k x_i = 0\}$ and $B'(x) = f(x) \cdot g(x)$, where $g(x)$ does not depend on $x_1$. This is possible since we may restrict our attention to those $x$ with $\bigoplus_{i=1}^k x_i = 0$ ($B'$ and $f$ are zero elsewhere), and for such values of $x$, the value of $x_1$ can be inferred from $x_2,\ldots, x_k$.
		
		Consider the expansion $B' = \sum_{S \seq [r]\times [m]} \alpha_S \chi_S$. One can easily verify that the Fourier expansion of $f$ is
		\[
		f = \sum_{U} \beta_U \chi_U, \qquad \mbox{where} \qquad \beta_U =
		\begin{cases}
		2^{-m}, & \qquad U=[k] \times U', U' \subseteq [m] \\
		0, & \qquad otherwise.
		\end{cases}
		\]
		Denote the Fourier expansion of $g$ by $g = \sum_{V \seq ([r]\sm \{1\}) \times [m]} \gamma_V \chi_V$.

		\medskip		
		
		Since $B'=f \cdot g$, for any $S$ we have $\alpha_S = \sum_{U \triangle V = S} \beta_U \gamma_V$. We claim that for any $S$, this sum consists of a single term, that is, $\alpha_S = \beta_{U_0} \gamma_{V_0}$ for some unique pair $(U_0,V_0)$ with $S = U_0 \triangle V_0$.
		To see this, assume $U \triangle V=S$ and $\beta_U, \gamma_V \neq 0$. Since $g$ does not depend on $x_1$, $V$ does not contain elements of the form $(1, i)$. Hence, given $S$ we may decode $U$ as $U = [k] \times U'$, where $U' = \{i: (1,i) \in S\}$ (recall $U$ is Cartesian, as otherwise $\beta_U=0$). Since we unambiguously determine $U$, we uniquely determine $V = S \triangle U$.
		
		\medskip
		
		Let us now compute $\norm{(B')^C}_2^2 = \sum_{S\text{ Cartesian}} \alpha_S ^ 2$.
		Write each $\alpha_S$ as $\beta_U \gamma_V$. We split the total contribution of the terms $\alpha_S^2 = (\beta_U \gamma_V)^2$ to $\sum \alpha_S^2$ into two cases, depending on $V$.
		
		\medskip \noindent \emph{Case 1: $V = \es$}. Observe that
		\[
		\gamma_V = \be[g] = 2^m \mu.
		\]
		To see this, fix a value $x_2^0,x_3^0,\ldots,x_r^0$ and let $x_1^0$ be such that $\bigoplus_{i=1}^k x_i^0 = 0_m$. Since $g$ does not depend on $x_1$, then for any $x_1$, we have
		\[
		g(x_1,x_2^0,\ldots,x_r^0)=g(x_1^0,x_2^0,\ldots,x_r^0)=B'(x_1^0,x_2^0,\ldots,x_r^0).
		\]
		On the other hand, $B'(x_1,x_2^0,\ldots,x_r^0)=0$ for any $x_1 \neq x_1^0$ (as $B'(x)=0$ for every $x$ s.t.~$\bigoplus_{i=1}^k x_i \neq 0_m$). Hence, each value $x$ with $B'(x) \neq 0$ corresponds to exactly $2^m$ values $x'$ s.t.~$g(x')=B(x)$, obtained from $x$ by changing only the first coordinate. Consequently, $\be[g]=2^m \mu$.
		
		Thus, the total contribution of terms with $V=\es$ is at most
		\[
		\sum_U \beta_U^2 \cdot \gamma_V^2 = 2^m \mu^2,
		\]
		 using the fact that $\sum_U \beta_U^2 = \be[f^2] = \be[f] = 2^{-m}$ by Parseval's identity.
		
		\medskip \noindent \emph{Case 2: $V \neq \es$}. We claim that in this case, for any $V$, there are at most two $U$'s for which $V \triangle U$ is Cartesian.
		Indeed, let $i \in [r]$ be such that $V_i \defeq V \cap (\{i\} \times [m]) \neq \emptyset$. Note that $V \cap (\{1\} \times [m]) = \emptyset$. Hence, if $U = ([k] \times U') \neq \emptyset$, then $(V \triangle U) \cap (\{1\} \times [m]) = U'$. On the other hand,
		\[
		(V \triangle U) \cap (\{i\} \times [m]) =
		\begin{cases}
		      V_i \triangle U', &\qquad i \leq k, \\
		      V_i, &\qquad i>k.
		\end{cases}
		\]
		As $V \triangle U$ is Cartesian, we have $(V \triangle U) \cap (\{i\} \times [m]) = (V \triangle U) \cap (\{1\} \times [m]) = U'$. As $V_i \neq \emptyset$, this is possible only if $i>k$ and $V_i=U'$, that is, $U=[k] \times V_i$. In addition, $U=\emptyset$ is possible if $V$ is Cartesian.
		
		Thus, the total contribution of terms $\alpha_S^2 = (\beta_U \gamma_V)^2$ with $V \neq \es$, is at most (recall $|\beta_U| \leq 2^{-m}$)
		\[
		2 \cdot (2^{-m})^2 \sum_V \gamma_V^2 = 2^{1-2m} \be[g^2] \leq 2^{1-2m} \be[g] =2^{1-m} \mu.
		\]
		We conclude
		\[
			\norm{(B')^C}_2^2 = \sum_{S\text{ Cartesian}} \alpha_S^2 \leq 2^m \mu^2 + 2^{1-m}\mu.
		\]
		This completes the proof.
		\end{proof}

		\subsubsection{Wrapping up the proof of the obfuscation lemma}
		
		\begin{proof}[Proof of Lemma~\ref{lem:xs-bound}]
			
			Combining lemmas~\ref{lem:noise} and~\ref{lem:non-cartesian}, and using $\be[B']\leq 1/\binom{r}{k}$, we get
			\[
				\cov(B'(x), B'(y)) \leq 2^{-2t} / \binom{r}{k} + 2^{-t} \cdot 2^{m} / \binom{r}{k}^2 + 2^{-t} \cdot 2^{1-m}/ \binom{r}{k} .
			\]
			Using Lemma~\ref{lem:to-bounded}, we deduce
			\[
				\Pr[P^{-1}(B(P(x))) = Q^{-1}(B(Q(y)))] =
				\binom{r}{k} \be_{x,y}[B'(x) B'(y)] \leq
				2^{-2t} +  \frac{2^{m-t}}{\binom{r}{k}} + 2^{-t+1-m}.
			\]
			This completes the proof.
		\end{proof}

\section{Hardness of Dense \tops{$k$}-SUM}
\label{sec:k-SUM}

In this section we prove Theorem~\ref{thm:intro-main-informal-SUM}. The precise formulation of the theorem is as follows.
\begin{theorem}[Sparse to dense $k$-SUM reduction]
\label{thm:red-sum}
        Let $M,N$ be integers such that $\sqrt{N} \leq M \leq N$.
        Assume there is an algorithm for $(k,M, N^{1/k})$-SUM with
        success probability $\beta$ and expected running time $\mathcal{T}$.
		Then, there is an algorithm for $(k, N, N^{1/k})$-SUM
		with success probability at least $\frac{\beta^{4}}{(\log M)^2 \cdot k^{O(k)}}$ and expected running time at most $N/M \cdot ( \mathcal{T}+ \tilde{O}_k(N^{1/k}))$.
\end{theorem}

The proof is based on the modular $k$-SUM problem, which we call $k$-MSUM.
\begin{definition}[Average-case $k$-MSUM problem]
In the $(k,N,r)$-MSUM problem, the input consists of $r$ elements $\z_1,\ldots,\z_r$, each chosen independently and uniformly at random from $\integers_{N}$. The goal is to find a $k$-tuple $K = \{ i_1,\ldots,i_k\}$, such that $\sum_{j \in K} \z_j \bmod N = 0$.
\end{definition}

\noindent Informally, the proof consists of three stages.
\begin{enumerate}
	\item \textbf{Reduction to $k$-MSUM.} In Section~\ref{sec:sub:red-to-modular}, we show that for obtaining the reduction from the $(k,N,N^{1/k})$-SUM problem to the $(k,M,N^{1/k})$-SUM problem (i.e. proving Theorem~\ref{thm:red-sum}), it is sufficient to devise a reduction from $(k,pq,r=(pq)^{1/k})$-MSUM to $(k,p,r)$-MSUM for a pair of primes $p,q$ that satisfy $pq \approx N$ and $p \approx M$.
	
	\item \textbf{Obfuscation process.} In Section~\ref{sec:sub:sum-obfuscation-process}, we introduce an obfuscation process that transforms $r$-tuples of vectors in $\Z_{pq}$ to $r$-tuples of vectors in $\Z_p$, similarly to the reduction for the $k$-XOR problem presented in Section~\ref{sec:kxor}. Specifically, we show that it is sufficient to prove an \emph{obfuscation lemma} which asserts that
the outputs of the $(k,p,r)$-MSUM oracle are sufficiently independent when it is applied to the `obfuscated' inputs.
	
	\item \textbf{Proof of the obfuscation lemma.} In Section~\ref{sec:sub:sum-obfuscation-proof} we prove the obfuscation lemma by employing techniques from discrete Fourier analysis and combinatorics. It is the most involved part of the proof of the main theorem.
\end{enumerate}

\subsection{Reduction to modular \tops{$k$}-SUM and proof of Theorem~\ref{thm:red-sum}}
\label{sec:sub:red-to-modular}

The main ingredient in the proof of Theorem~\ref{thm:red-sum} is the following lemma, which provides reduction from $k$-MSUM modulo $pq$ to $k$-MSUM modulo $p$, for prime numbers $p,q$ with $p > q$.
\begin{lemma}[Main $k$-MSUM reduction]\label{lem:red-sum}
    Let $p,q,r$ be positive integers such that $p > q$ are prime numbers and $pq \leq \binom{r}{k}$.
    Assume there is an algorithm for $(k,p,r)$-MSUM with success
    probability $\beta$ and expected running time $\mathcal{T}$.

	Then, there is an algorithm for $(k, pq, r)$-MSUM with success probability
	\[
		\Omega \left( \Big( \frac{\beta^{2} \cdot pq }{k \log(q) \cdot \binom{r}{k}} \Big) ^ 2\right)
	\]
	and expected running time at most $q \cdot ( \mathcal{T}+ \tilde{O}(r) )$.
\end{lemma}
In this subsection we prove that Lemma~\ref{lem:red-sum} implies Theorem~\ref{thm:red-sum}. The (more complex) proof of Lemma~\ref{lem:red-sum} spans the following subsections.

The derivation of Theorem~\ref{thm:red-sum} from Lemma~\ref{lem:red-sum} relies on two additional reductions.
\begin{lemma}[Simple reduction]
\label{lem:sparsesum}
    Let $r,r'$ be positive integers such that $r \geq r'$.
    Assume there is an algorithm for $(k,N, r)$-SUM (resp., MSUM), with success
    probability $\beta$ and expected running time $\mathcal{T}$.
	Then, there is an algorithm for $(k, N, r')$-SUM (resp., MSUM) with success probability at least
   $ \frac{\beta}{(2r/r')^k}$ and expected running time at most $\mathcal{T} + \tilde{O}_k(r)$.
\end{lemma}
The proof of this lemma is essentially the same as the proof of the analogous Lemma~\ref{lem:sparse} for $k$-XOR,
and we omit it.

\begin{lemma}[$k$-SUM to $k$-MSUM reduction]
\label{lem:to-mod}
Let $M,N$ be positive integers such that $M \leq 2N+1$.
Assume there is an algorithm for $(k,M,r)$-MSUM with
success probability $\beta$ and expected running time $\mathcal{T}$.
Then, there is an algorithm for $(k,N,r)$-SUM
with success probability $\Omega \left(\frac{\beta}{k(8N \cdot k/M)^k}\right)$ and expected running time at most $\mathcal{T} + \tilde{O}_k(r)$.
\end{lemma}

\begin{proof}
Denote $r' = \lceil r \cdot M/4N \rceil$.
Given an algorithm for $(k,M,r)$-MSUM with success probability $\beta$ and running time $\mathcal{T}$,
by Lemma~\ref{lem:sparsesum}, there is an algorithm $B$ for $(k,M,r')$-MSUM
with success probability at least $\frac{\beta}{(8N/M)^k}$ and expected running time at most
$\mathcal{T} + \tilde{O}_k(r)$. We use $B$ to devise an algorithm $A$ for $(k,N,r)$-SUM.

We first assume that $M$ is odd.
On input that consists of $r$ integers $\z_1,\ldots,\z_r$ uniform in $\{-N,\ldots,N\}$, $A$ performs
the following steps.
\begin{enumerate}
  \item Discard all $\z_i$ such that $\z_i \notin \{-(M-1)/2 ,\ldots,(M-1)/2\}$.
    Denote the number of remaining elements by $r_1$. If $r_1 < r'$, then return failure.
    Otherwise, take the first $r'$ remaining elements and denote them by $u_1,\ldots,u_{r'}$.
  \item Define the mapping $u_i \mapsto y_i$ (onto $\integers_{M}$) by $y_i = u_i \bmod M$.
        Note that $y_1,\ldots,y_{r'}$ is a $(k,M,r')$-MSUM instance.
  \item Pick $k$ elements $v'_1,\ldots,v'_k$ in $\integers_M$ uniformly at random, conditioned
        on $\sum_{j \in [k]} v'_j \bmod M = 0$.
        Then, for each $i \in [r']$, pick $j \in [k]$ uniformly at random and define $v_i = v'_j$
        and $x_i = y_i + v_i \bmod M$.
  \item Run $B$ on the input $x_1,\ldots,x_{r'}$ and assume it returns
        a $k$-tuple $K'$ such that $\sum_{i \in K'} x_{i} \bmod M = 0$.
        Trace $K'$ back to the corresponding $k$-tuple $K$ for $\z_1,\ldots,\z_r$ and if
        $\sum_{i \in K} \z_{i} = 0$, return $K$. Otherwise, return failure.
\end{enumerate}
Note that we do not run $B$ directly on $y_1,\ldots,y_{r'}$
since it may be malicious and return
$k$-tuples that sum to 0 modulo $M$, but never give a $k$-SUM over the integers for $\z_1,\ldots,\z_r$.

\paragraph{Analysis.}
Clearly, the algorithm returns a correct output if it succeeds and its complexity is as claimed.
To analyze the success probability, we consider the following events:
	\begin{enumerate}
        \item $r_1 \geq r'$.
		\item $B$ returns a $k$-tuple $K'$ such that $\sum_{i \in K'} x_{i} \bmod M = 0$.
		\item $\{v_i\}_{i \in K'} = \{v'_j\}_{j \in [k]}$ (as possible multi-sets). Note that this implies  $\sum_{i \in K'} v_{i} \bmod M =
        \sum_{j \in [k]} v'_j \bmod M = 0$ and therefore,
        $$\sum_{i \in K'} u_{i} \bmod M = \sum_{i \in K'} y_{i} \bmod M = \sum_{i \in K'} (x_{i} - v_i) \bmod M = 0.$$
        \item $\sum_{i \in K'} u_{i} = 0$.
	\end{enumerate}
Observe that if the fourth event occurs, then $\sum_{i \in K} \z_{i} = \sum_{i \in K'} u_{i} = 0$ and
$A$ succeeds. In the following, we lower bound the probability of these events.

First, $\be[r_1] = r \cdot M/(2N + 1)$ and a simple tail bound gives $\Pr[r_1 \geq r'] \geq 1/4$ for $r$ sufficiently large (i.e., larger than some constant value).
Conditioned on the first event, the second event occurs with probability at least $\frac{\beta}{(8N/M)^k}$.
The third event occurs with probability at least $1/(k!) > k^{-k}$.
Note that since $v_1,\ldots,v_{r'}$ are picked independently of $y_1,\ldots,y_{r'}$,
then the third event is independent of the second.
Thus, the first three events occur with probability
$\Omega \left(\frac{\beta}{(8N \cdot k/M)^k}\right)$.

Finally, recall that $v'_1,\ldots,v'_k$ are picked uniformly at random from $\integers_M$, conditioned
on $\sum_{j \in [k]} v'_j \bmod M = 0$.
Thus, conditioning on $\sum_{i \in K'} x_{i} \bmod M = 0$,
and on the event that $\{v_i\}_{i \in K'} = \{v'_j\}_{j \in [k]}$
(but not on the individual values of each $v'_{j}$),
the $k$-tuple $\{u_{i}\}_{i \in K'}$
is uniformly distributed in $\{-(M-1)/2 ,\ldots,(M-1)/2\}^k$,
conditioned on $\sum_{i \in K'} u_{i} \bmod M = 0$.
Given this distribution of $\{u_{i}\}_{i \in K'}$, it remains to lower bound
the probability that $\sum_{i \in K'} u_{i} = 0$ by $1/k$.

Write $U = \sum_{i \in K'} u_{i}$, as a sum of $k$ uniform integers in $\{-(M-1)/2 ,\ldots,(M-1)/2\}$.
Note that for any $t$,
\begin{align}\label{eq:convexity}
\begin{split}
	\Pr[U = t] \leq & \Pr[U = t+1] \quad \text{if} \quad t+1 \leq \be[U], \\
	\Pr[U = t] \geq & \Pr[U = t+1] \quad \text{if} \quad t \geq \be[U].
\end{split}
\end{align}
To see this, observe that the function $t \mapsto \Pr[U = t]$ is log-concave, as a convolution of (discrete) log-concave functions (see for example~\cite[Proposition 10\textit{(vii)}]{MY96}), and is symmetric around $t = \be[U]$.

Since in our case $\be[U] = 0$, then for any $t$, we have $\Pr[U = 0] \geq \Pr[U = t]$.
Hence,
\begin{align*}
\Pr \li[ \given{U=0}{U \bmod M = 0} \ri] = & \frac{\Pr[U = 0]}{\Pr[U \bmod M = 0]} \geq
\frac{\Pr[U = t]}{\Pr[U \bmod M = 0]} \\ \geq & \Pr \li[ \given{U=t}{U \bmod M = 0} \ri].
\end{align*}
As $U \in \{-k(M-1)/2 ,\ldots,k(M-1)/2\}$, given that $U \bmod M = 0$, $U$ can only attain $k$ possible values,
implying that $\Pr \li[ \given{U=0}{U \bmod M = 0} \ri] \geq 1/k$.

Finally, if $M$ is even, we change the algorithm to
remove $\z_i \notin \{-M/2 ,\ldots,M/2 - 1\}$.
The analysis is similar, but we have $\be[U] = -k/2$.
Nevertheless, the final result is unchanged since
$\Pr[U = 0] \geq \Pr[U = t]$ for every $t$ such that $t \bmod M = 0$ (assuming $M$ is larger than $k$; we indeed may assume $M \geq k$ as otherwise $(k,N,r)$-SUM can be solved with the poor probability $1/N^k$ by searching for $k$ zeros in the input).
\end{proof}

We now derive Theorem~\ref{thm:red-sum} from Lemmas~\ref{lem:red-sum},~\ref{lem:sparsesum} and~\ref{lem:to-mod}.
\begin{proof}[Proof of Theorem~\ref{thm:red-sum}]
Let $M,N$ be such that $\sqrt{N} \leq M \leq N$.
Our goal is to devise an algorithm for $(k, N, N^{1/k})$-SUM,
given an algorithm $B$ for $(k, M, N^{1/k})$-SUM with
success probability $\beta$ and expected running time $\mathcal{T}$.

Clearly, $B$ can be applied to solve
$(k,M,N^{1/k} \cdot 2k)$-SUM with the same success probability and complexity.

Let $p$ be a prime number that satisfies $M \leq p < 2M$.\footnote{Such a prime can clearly be found efficiently since by a quantitative version of the prime number theorem, the number of primes between $M$ and $2M$ is $\Omega(M/\log M)$. Therefore, one may pick such a prime by random sampling and using a standard primarily test algorithm.}
By Lemma~\ref{lem:sparsesum}, based on $B$, there is an algorithm $B_1$
for $(k,(p-1)/2,N^{1/k} \cdot 2k)$-SUM
with success probability at least $\beta_1 = \frac{\beta}{k^{O(k)}}$ and
expected running time at most $\mathcal{T} + \tilde{O}_k(N^{1/k})$.
$B_1$ immediately gives an algorithm for $(k,p,N^{1/k} \cdot 2k)$-MSUM with the same parameters
(as a $K$-tuple that sums to 0 over the integers sums to zero $\bmod$ $p$).

Let $q$ be a prime number such that $N/2p \leq q < N/p$.
Note that we have
\[
pq < N < \binom{N^{1/k} \cdot 2k}{k} < N \cdot k^k \qquad \mbox{and} \qquad
q < N/p \leq N/M \leq M \leq p.
\]
Hence, we can apply Lemma~\ref{lem:red-sum} based on $B_1$
to deduce that there is an algorithm $B_2$ for
$(k,pq,N^{1/k} \cdot 2k)$-MSUM
with success probability at least
\[
\beta_2 = \Omega \left(\frac{\beta_1^{4} (pq)^2}{k^2 \log(q)^2 \binom{N^{1/k}\cdot 2k}{k}^2}\right) \geq \frac{\beta^{4}}{(\log M)^2 \cdot k^{O(k)}}
\]
and expected running time at most
$q \cdot ( \mathcal{T}+ \tilde{O}_k(N^{1/k})) \leq N/M \cdot ( \mathcal{T}+ \tilde{O}_k(N^{1/k}))$.

Noting that $N/2 \leq pq \leq N$,
we invoke Lemma~\ref{lem:to-mod} based on $B_2$, and conclude that there is an algorithm $B_3$
for $(k,N,N^{1/k} \cdot 2k)$-SUM
with success probability at least $\beta_3 = \frac{\beta_2}{k^{O(k)}} = \frac{\beta^{4}}{(\log M)^2 \cdot k^{O(k)}}$
and expected running time at most
$N/M \cdot ( \mathcal{T}+ \tilde{O}_k(N^{1/k}))$.

Finally, we apply Lemma~\ref{lem:sparsesum} based on $B_3$
and deduce that there is an algorithm for $(k,N,N^{1/k})$-SUM
with success probability at least $\frac{\beta_3}{k^{O(k)}} = \frac{\beta^{4}}{(\log M)^2 \cdot k^{O(k)}}$
and expected running time at most
$N/M \cdot ( \mathcal{T}+ \tilde{O}_k(N^{1/k}))$.
\end{proof}

\subsection{The obfuscation process}
\label{sec:sub:sum-obfuscation-process}

Our goal in the rest of this section is to prove Lemma~\ref{lem:red-sum}. The proof strategy is similar to the proof of Lemma~\ref{lem:mainxor} in the $k$-XOR case presented in Section~\ref{sec:kxor}. Namely, we devise an algorithm that receives $r$ vectors in $\Z_{pq}$, denoted by $\z_1,\ldots,\z_r$, and randomly obfuscates them, returning $r$ vectors in $\Z_p$, denoted by $y_1,\ldots,y_r$. The main properties of the obfuscation are that a solution to the $k$-MSUM $y$-problem gives rise to a solution of the $k$-MSUM $\z$-problem with a good probability (i.e., $\approx 1/q$), and that the applications of the oracle are sufficiently independent to yield a solution of the $\z$-problem with the desired probability.

In this subsection we present the obfuscation algorithm, state the main obfuscation lemma which asserts that our algorithm achieves its goals, and derive Lemma~\ref{lem:red-sum} from the obfuscation lemma. The proof of the obfuscation lemma is presented in the next subsection.

\subsubsection{The obfuscation algorithm and the obfuscation lemma}

\medskip Let $p,q,r$ be positive integers such that $p,q$ are prime numbers, $p \geq q$,  and $pq \leq \binom{r}{k}$.
Let $B$ be an algorithm for $(k,p,r)$-MSUM. Let $L$ be a parameter to be specified below.
We define the algorithm $A$ for $(k,pq,r)$-MSUM, which receives as an input an $r$-tuple $(\z_1,\ldots,\z_r) \in \integers_{pq}^r$ of elements in $\integers_{pq}$, as follows.

\begin{alg}\label{alg:sum}\skipline
\begin{enumerate}[noitemsep]
	\item
	\textbf{Repeat} $L$ times:
	\item
	\begin{adjustwidth}{5mm}{}
	\textbf{Draw} uniformly random invertible $\alpha \sim \integers_{p\cdot q}^\ast$ and $\gamma \sim \integers_p^\ast$ and a uniformly random
    permutation $P \in S_r$.
	\end{adjustwidth}
	\item
	\begin{adjustwidth}{5mm}{}
	\textbf{Let} $y_{i} = \gamma \cdot \lfloor (\alpha \cdot \z_{P(i)} \modd pq)  / q \rceil \modd p$ for all $i \in [r]$.
	\end{adjustwidth}
	\item
	\begin{adjustwidth}{5mm}{}
	\textbf{Feed} $B$ with $(y_1, \ldots, y_r)$. In case it outputs a $k$-tuple $K$ with $\sum_{i \in K} y_i \% p = 0$,
    \textbf{test} whether $\mathcal{K} = P(K)$ satisfies $\sum_{i \in \mathcal{K}} \z_{i} \% pq = 0$, and \textbf{if} it does -- \textbf{output}
    the $k$-tuple $\mathcal{K}$. \textbf{Otherwise, continue}.
	\end{adjustwidth}
\end{enumerate}
\end{alg}

The obfuscation lemma in the arithmetic case is as follows.
\begin{lemma}\label{lem:ss-bound}
	Let $p,q,r$ be positive integers such that $p \geq q$ are prime numbers and $pq \leq \binom{r}{k}$.

	Let $(\z_1,\ldots,\z_r) \in \integers_{pq}^r$ be chosen uniformly at random. Let the $r$-tuples
	$(y_1^{(1)},\ldots, y_r^{(1)}) \in \integers_{p}^r$ and $(y_1^{(2)},\ldots,y_r^{(2)}) \in \integers_{p}^r$
	be obtained from it by the procedure described above (in two out of the $L$ iterations).
	Let $\mathcal{K}_1, \mathcal{K}_2$ be the two corresponding $\mathcal{K}$'s obtained in the process.
	Then,
	\begin{equation}\label{eq:indep2}
		\Pr[\mathcal{K}_1 = \mathcal{K}_2] \leq O(\log(q) / q^2),
	\end{equation}
	where the probability is taken over $z,y^{(1)}, y^{(2)}$, and $B$'s randomness.
\end{lemma}

Note that cases where at least one of $\mathcal{K}_1,\mathcal{K}_2$ is not obtained (that is, when in at least one of the two iterations, Algorithm~$B$ fails to find a solution to the $(k,p,r)$-MSUM problem) are not counted as equality between $\mathcal{K}_1$ and $\mathcal{K}_2$.

\subsubsection{Proof of the main reduction lemma}

We now deduce Lemma~\ref{lem:red-sum} from the obfuscation lemma (Lemma~\ref{lem:ss-bound}).

\begin{proof}[Proof of Lemma~\ref{lem:red-sum}]
As in the corresponding proof of Lemma~\ref{lem:mainxor} for $k$-XOR, we analyze a tweaked version of the algorithm (with the same success probability) in which all $L$ iterations are performed. For any $1 \leq l \leq L$, let $\mathcal{K}_l$ be the $\mathcal{K}$ obtained in the $l$'s iteration ($\mathcal{K}_l$ exists only when B succeeds, i.e., with probability $\beta$).
		Denote by $S_l$ the event that $\mathcal{K}_l$ admits a solution to the $(k, pq, r)$-MSUM problem.
		We claim that for all $l=1,\ldots,L$,
		\[
			\Pr[ S_{l} ] \geq \beta \cdot \Omega(1/(\sqrt{k} q)).
		\]
		Indeed, $S_l$ occurs if the $B$ oracle succeeds, and in addition,  $\sum_{i \in \mathcal{K}} \z_i \modd pq = 0$ holds, given that $\sum_{i \in K} y_i \modd p = 0$.
		
		The probability of the first event is $\beta$, as $y_1, \ldots, y_r$ are uniformly and independently distributed in $\integers_{p}$. (Indeed, conditioning on the variables $\alpha, P$ (but not on $\z_{P(i)}$), the variables $\wt{y}_i = \lfloor (\alpha \cdot \z_{P(i)} \modd pq) / q \rceil \modd p$ are uniformly and independently distributed in $\integers_{p}$, and so are $y_i=\gamma \cdot \wt{y}_i$.)

		To see that the probability of the second event is $\Omega(1/(\sqrt{k} q))$, notice that $\alpha \z_{P(i)} \modd pq = (q\gamma^{-1} y_i + \sigma_i) \modd pq$, with $\sigma_i \in \{(1-q)/2, \ldots, (q-1)/2\}$. Conditioning on $\alpha, \gamma, P, y_i$ (and not on $\z_{P(i)}$), $\sigma_i$ is uniformly distributed in this set. Observe that given $\sum_{i \in K} y_{i} \modd p = 0$, the event $\sum_{i \in K} \alpha \z_{P(i)} \modd pq = 0$ is equivalent to $\sum_{i \in K} \sigma_i \modd pq = 0$. The probability of this latter event is $\Omega(1/(\sqrt{k} q))$, for example by~\eqref{eq:convexity} and Chebyshev's inequality with the standard deviation of $\sum_{i \in K} \sigma_i$ equal to $\Theta(\sqrt{k} q)$.

		Define the random variables
		\[
		Z'
		\defeq
		\sum_{l=1}^{L} \one\{S_l\} -\sum_{1 \leq l < l' \leq L} \one\{\mathcal{K}_l = \mathcal{K}_{l'}\} , \qquad Z \defeq \max(Z', 0).
		\]
		Similarly to the proof of Lemma~\ref{lem:mainxor} in Section~\ref{sec:mainred}, it is easy to verify that (tweaked) Algorithm~\ref{alg:sum} succeeds to solve the $(k, pq, r)$-MSUM problem with probability at least
		\[
			\pr[Z > 0] \geq \frac{\be[Z]^2}{\be[Z^2]}.
		\]
		To bound $\be[Z]$ from below, note that
		\[
		\be[Z] \geq \be[Z'] = \sum_{l=1}^{L} \pr[S_l] - \sum_{1 \leq l < l' \leq L} \pr[\mathcal{K}_l = \mathcal{K}_{l'}].
		\]
		Using $\Pr[S_l] \geq \beta \cdot \Omega(1/(\sqrt{k} q))$ and Lemma~\ref{lem:ss-bound}, we get
		\[
		\be[Z] \geq L\cdot \Omega(\beta / (\sqrt{k} q)) - \binom{L}{2} O(\log(q) / q^2).
		\]
		We choose $L= c \cdot \beta q / (\sqrt{k} \log(q))$ for a sufficiently small constant $c$, and obtain
		\[
			\be[Z]\geq \Omega\Big(\frac{\beta^2}{k\log(q)}\Big).
		\]
		To bound $\be[Z^2]$ from above, note that similarly to the proof of Lemma~\ref{lem:mainxor}, we have
		\[
			\be[Z^2] \leq \frac{2}{(pq)^2} \binom{r}{k}^2.
		\]
		Therefore, (tweaked) Algorithm~\ref{alg:sum} succeeds with probability at least
		\[
			\pr[Z > 0] \geq \Omega(\beta^4) \cdot \Big(\frac{pq}{\binom{r}{k} k \log(q) }\Big)^2.
		\]
		The running time of the algorithm is
		\[
			L \cdot ( \mathcal{T}+ \tilde{O}(r) ) \leq q \cdot ( \mathcal{T}+ \tilde{O}(r) ).
		\]
		This completes the proof of the lemma.
\end{proof}

\subsection{Proof of the obfuscation lemma}
\label{sec:sub:sum-obfuscation-proof}

In this section we prove Lemma~\ref{lem:ss-bound}. We start by introducing a distribution that models two independent outputs of the obfuscation process, and restate the obfuscation lemma.

	\begin{definition}\label{def:noise2}
	Let $p, q$ be prime numbers.
	We say that a pair of random variables $(x^{(1)},x^{(2)})$, each taking values in $\integers_{p}^r$, has a $(p,q,r)$-arithmetic-distribution, if there exist random variables
	$\z,\, \alpha^{(j)},\, \gamma^{(j)}$, $j=1,2$,
	with:
	\begin{enumerate}[noitemsep]
		\item $\z,\, \alpha^{(1)},\, \alpha^{(2)},\, \gamma^{(1)},\, \gamma^{(2)}$,
		are independent random variables.

		\item $\z \sim \integers_{p\cdot q}^r$ is uniformly distributed.

		\item $\alpha^{(j)} \sim \integers_{p\cdot q}^{\ast}$ is a uniformly random invertible element of $\integers_{p\cdot q}$.

		\item $\gamma^{(j)} \sim \integers_{p}^{\ast}$ is a uniformly random nonzero residue modulo $p$.

		\item For all $i=1\ldots r$, $j=1,2$, we have
		\[
			x^{(j)}_i =	\gamma^{(j)} \cdot \lfloor \alpha^{(j)} \cdot \z_i / q \rceil \modd p.
		\]
	\end{enumerate}
	\end{definition}

	\begin{lemma}
		\label{lem:aux-ss-bound}
		Let $B$ be an algorithm which receives as input a list of $r$ integers in $\integers_{p}$, and outputs the indices of $k>0$ numbers among them whose SUM is $0$ (modulo $p$).
		If $(x,y)$ has a $(p,q,r)$-arithmetic-distribution with $p \geq q$, and $P,Q\sim S_r$ are two uniformly random and independent permutations, then
		\begin{equation}\label{eq:ss-bound}
			\Pr[P^{-1}(B(P(x))) = Q^{-1}(B(Q(y)))] \leq O\Big( \frac{\log(q)}{q^2} + \frac{p}{q\binom{r}{k}} \Big),
		\end{equation}
		where the probability is taken over $B$'s randomness, $x,y$ and $P,Q$.
		(The event on the left hand
		side is contained in the event that both executions $B(P(x))$, $B(Q(y))$ succeed).
	\end{lemma}
	Notice that Lemma~\ref{lem:ss-bound} immediately follows from Lemma~\ref{lem:aux-ss-bound} (compare~\eqref{eq:indep2} with~\eqref{eq:ss-bound}).

\subsubsection{Proof outline} The proof of Lemma~\ref{lem:aux-ss-bound} uses techniques from \emph{discrete Fourier analysis} and combinatorial techniques.
Technically, the proof is more complicated than for $k$-XOR since the $k$-SUM analog of Lemma~\ref{lem:noise} does not hold. Instead, it is replaced by the sequence of lemmas~\ref{lem:noise2}--\ref{lem:mag-exp}. The proof consists of three steps presented in the order of appearance in the paper.
\begin{enumerate}
	\item \textbf{Transformation to real-valued functions.}
	Similarly to the $k$-XOR obfuscation, we show that instead of analyzing the obfuscation on a tuple-valued function, it is sufficient to analyze its action on the simpler class of real-valued functions.
	We utilize the fact that our obfuscation randomly permutes the input vectors, so that any oracle $B:\integers_{p}^r \to {\binom{[r]}{k}}$
must, informally, treat all candidate output $k$-tuples in the same way.
	Hence it suffices to analyze the modified, real-valued, oracle $B':\integers_{p}^r \to [0,1]$ which essentially indicates the probability that $B$ outputs the specific $k$-tuple $K \defeq \{1,\ldots,k\}$ when applied on its input. Specifically, our task is reduced to showing that
	\begin{equation}\label{eq:intro-to-bounded-sum}
	 	\be[B'(x) B'(y)] \leq O\Big( \frac{\log(q)}{q^2} + \frac{p}{q\binom{r}{k}} \Big) / r^k,
	 \end{equation}
	where $x,y$ are two independent obfuscations of a common, random, $z \in \integers_{pq}^r$.

	\item \textbf{Representing the correlation in terms of the Fourier expansion.} In order to prove~\eqref{eq:intro-to-bounded-sum}, we consider the Fourier expansion of $B'$, namely,
	\[
		B'=\sum_{S \in \integers_{p}^r} \wh{B}'(S) \chi_S, \qquad \mrm{where} \qquad \chi_S(v) = \exp \left(\frac{2\pi i}{p}\langle S,v \rangle \right), \qquad \wh{B'}(S) \in \mathbb{C}.
	\]

	It turns out that the correlation between two obfuscations (which appear in different iterations of Algorithm~A described above) is a weighted sum of the squared Fourier coefficients $\wh{B}'(S)^2$:
	\begin{equation}\label{eq:summary-bound1}
		\mathrm{Cov}(B'(x),B'(y))=\sum_{S \neq 0} (p-1)M_{p,q,r}(S)\wh{B}'(S)^2,
	\end{equation}
	where $M_{p,q,r}(S)$ roughly serve as the `weight' for $\wh{B}'(S)^2$, and are defined as $M_{p,q,r}(S) = \be[\chi_S(x) \overline{\chi_S(y)}]$.
		
	\item \textbf{Bounding the correlation using discrete Fourier analysis.}
	We bound the correlation, the right hand side of~\eqref{eq:summary-bound1}, in two steps:
		\begin{enumerate}
			\item \emph{Partitioning into 2-dimensional subspaces.} We use the structure of solutions to the $k$-SUM problem to show that it is sufficient to bound $\sum_{S \in U} M_{p,q,r}(S)$ over two-dimensional subspaces $U = \{a S_1 + b S_2:a,b \in \integers_{p}\} \seq \integers_p^r$ of a certain kind.
			
			\item \emph{Bounding $\sum_S M_{p,q,r}(S)$ over 2-dimensional subspaces.} We bound $\sum_{S \in U} M_{p,q,r}(S)$ over two-dimensional subspaces $U$, using a combinatorial approach. Specifically, we represent this sum as the bias introduced in an event, caused by the dependence between $x,y$. We show that this bias is related to a Littlewood-Offord-type problem~\cite{LO43}.
			Specifically, in Lemma~\ref{lem:mag5} we relate $M_{p,q,r}(S)$ to the probability $\pr[\langle S, u\rangle = 0]$ for a random vector $u \in \integers_p^r$. This probability concerns the event where the sum of random variables $\langle S, u\rangle = \sum_i S_i u_i$ equals $0$. Bounding this type of probability is common in Littlewood-Offord theory. To this end, we use a classical antichain argument (see~\cite{Erdos45}) along with simple number theoretic estimates.
		\end{enumerate}
\end{enumerate}

\subsubsection{Transformation to real-valued functions}

		\begin{lemma}\label{lem:to-bounded-sum}
			Let $\func{B}{\integers_p^{r}}{\binom{[r]}{k}}$ be an algorithm that outputs either a $k$-tuple $R$ with $\sum_{i \in R} x_i \% p = 0$, or a failure message.
			Let $K \defeq \{1,\ldots , k\}$ and define $\func{B'}{\integers_p^{r}}{[0,1]}$ by
			\begin{equation}\label{eq:2_2}
				B'(x) = \Pr_{P,\gamma}[B(P(\gamma \cdot x))=P(\{1,2,\ldots,k\})],
			\end{equation}
			where $P \sim S_r$ is a uniformly random permutation, and $\gamma \sim \integers_{p}^*$ is a uniformly random invertible element of $\integers_p$.
			Then,
			\begin{equation}\label{eq:sum-annihilate}
				\sum_{i \in K} x_i \neq 0 \  \implies \  B'(x) = 0,
			\end{equation}
			\begin{equation}\label{eq:sum-measure-bound}
				\mu \defeq \be_x[B'(x)] \leq 1/\binom{r}{k},
			\end{equation}
			and if $(x,y)$ has a $(p,q,r)$-arithmetic distribution and $P,Q \sim S_r$ are independent, then
			\begin{equation}\label{eq:sum-to-bounded-need}
				 \Pr[P^{-1}(B(P(x))) = Q^{-1}(B(Q(y)))] = \binom{r}{k} \be_{x,y}[B'(x) B'(y)].
			\end{equation}
		\end{lemma}
We note that the probability in~\eqref{eq:2_2} is also taken over $B$'s randomness.
The proof of the lemma is similar to the proof of Lemma~\ref{lem:to-bounded}.

		\begin{proof}
			To show~\eqref{eq:sum-annihilate} note that if $\sum_{i \in K} x_i \neq 0$ then $B$ cannot output $P(K)$ on the input $P(\gamma \cdot x)$, by our assumption on $B$, and $\gamma$ being invertible. Hence $B'(x) = 0$ in such a case.

			To verify~\eqref{eq:sum-measure-bound}, denote $x' = P(\gamma \cdot x)$ and observe that when $x \sim \integers_{p}^r$, we have  $x' \sim \integers_{p}^r$ independently of $P$. Hence, by interchanging order of summation,
			\begin{align*}
				\be_{x}[B'(x)] &= \be_{P,\gamma}[\be_{x}[\one\{B(P(\gamma \cdot x)) = P(K)\}]] = \be_{P,\gamma}[\be_{x'}[\one\{B(x') = P(K)\}]] \\
				&=  \be_{x'}[\be_{P}[\one\{B(x') = P(K)\}]] \leq 1/\binom{r}{k},
			\end{align*}
			where the latter inequality holds because for any fixed $x'$, $P(K)$ attains the value of $B(x')$ with probability at most $1/\binom{r}{k}$.

			In order to prove~\eqref{eq:sum-to-bounded-need}, we reason about $\be_{x,y}[B'(x) B'(y)]$. Observe that for any $K' \seq [r]$ with $|K'| = k$, the function $B'_{K'}$ defined by $B'_{K'}(x) = \be_{P,\gamma}[\one\{B(P(\gamma \cdot x)) = P(K')\}]$ satisfies
			\begin{equation}\label{eq:sum-3}
			\be_{x,y}[B'_{K'}(x) B'_{K'}(y)] =  \be_{x,y}[B'(x) B'(y)].
			\end{equation}
			Indeed, let $R \in S_r$ be such that $R(K) = K'$. As $(R(x), R(y))$ has the same distribution as $(x, y)$, we have
			\[
			\begin{aligned}
				\mathbb{E}_{x,y}[B'_{K'}&(x) B'_{K'}(y)]
				= \mathbb{E}_{x,y}[B'_{K'}(R(x)) B'_{K'}(R(y))] \\
				&=
				\be_{x,y}\big[\be_{P', \gamma'}[\one\{B(P'R(\gamma' \cdot x)) = P'(K')\}] \be_{P'', \gamma''}[\one\{B(P''R(\gamma'' \cdot y))=P''(K')\}]\big] \\
				&= \be_{x,y}\big[\be_{P',\gamma'}[\one\{B(P'(\gamma' \cdot x))= P'R^{-1}(K')\}] \be_{P'', \gamma''}[\one\{B(P''(\gamma \cdot y))=P''R^{-1}(K')\}]\big] \\
				&=
				\be_{x,y}[B'(x) B'(y)].
			\end{aligned}
			\]
			Notice that if $(x,y)$ has a $(p,q,r)$-arithmetic distribution, and $\gamma',\gamma'' \sim \integers_{p}^*$ are uniformly random invertible elements of $\integers_p$ independent of all other variables, then $(\gamma' \cdot x, \gamma'' \cdot y)$ has a $(p,q,r)$-arithmetic distribution as well. We verify~\eqref{eq:sum-to-bounded-need}:
			\[
			\begin{aligned}
				\Pr_{x,y,P,Q}& \big[P^{-1}(B(P(x))) =
				Q^{-1}(B(Q(y)))\big] \\
				&= \sum_{K'}  \be_{x,y,P,Q}\big[\one\{P^{-1}(B(P(x)))=K'\} \one\{Q^{-1}(B(Q(y)))=K'\} \big] \\
				&= \sum_{K'} \be_{x,y,P,Q,\gamma',\gamma''}\big[\one\{P^{-1}(B(P(\gamma' \cdot x)))=K'\} \one\{Q^{-1}(B(Q(\gamma'' \cdot y)))=K'\} \big] \\
				&= \sum_{K'} \be_{x,y}\big[B'_{K'}(x) B'_{K'}(y)\big] \\
				&= \binom{r}{k} \be_{x,y} \big[ B'(x) B'(y)  \big],
			\end{aligned}
			\]
			where 
			the ultimate equality uses~\eqref{eq:sum-3}.
		\end{proof}

\medskip Due to the structure of the $k$-SUM problem, the function $B'$ has several properties that will be crucially used in the sequel. Before stating these properties, let us introduce Fourier expansion over $\integers_p$.

	\begin{definition}[Fourier expansion]
		Given $S \in \integers_{p}^r$, define $\func{\chi_{S}}{\integers_{p}^r}{\complex}$ by $\chi_S(x) = e_p(\langle S, x\rangle)$, where
		\[
			e_p(a) \defeq \exp\Big(\frac{2\pi i a}{p}\Big).
		\]
\end{definition}
		The set $\{\chi_S\}_{S \in \integers_{p}^r}$ is an orthonormal basis for the set of functions $\set{f}{\func{f}{\integers_{p}^r}{\complex}}$, with respect to the standard inner product $\langle f,g \rangle = \be_{x\sim \integers_{p}^r}[f(x) \overline{g(x)}]$. Hence, each $\func{f}{\integers_{p}^r}{\complex}$ can uniquely be decomposed as
		\[
			f = \sum_{S \in \integers_{p}^r} \wh{f}(S) \chi_S, \qquad \mrm{with} \qquad \wh{f}(S) \in \mathbb{C}.
		\]

\begin{claim}\label{claim:useful-sum}

Let $B'$ be defined as in Lemma~\ref{lem:to-bounded-sum}. Then:
\begin{enumerate}
\item For any $x \in \integers_p^r$ and any $\gamma \in \integers_p^*$, we have $B'(x)=B'(\gamma \cdot x)$.

\item $B'$ can be written in the form $B'(x)=I_k(x) \cdot g(x)$, where $I_k(x)=\one\{\sum_{i=1}^k x_i = 0\}$.

\item Let $S' \in \integers_p^r$ be defined by $S'_j =\one\{j \in [k]\}$. For any $S \in \integers_p^r$ we have $\wh{B'}(S) = \wh{B'}(S+S')$.
\end{enumerate}
\end{claim}  		
\begin{proof}

The first assertion holds trivially, by the definition of $B'$. The second holds since $B(x) = 0$ whenever $I_k(x) = 0$, with $I_k(x) = \one\{ \sum_{i=1}^{k} x_i = 0 \}$. Finally, the third holds as the Fourier expansion of $I_k$ is given by $I_k = \frac{1}{p}\sum_{\nu \in \integers_p} \chi_{\nu \cdot S'}$ with $S'_j = \one\{j \in [k]\}$.
\end{proof}

Following Claim~\ref{claim:useful-sum}, we shall study and exploit properties of functions $f:\integers_p^r \to [0,1]$ that satisfy $f(\gamma x)=f(x)$ and $\wh{f}(S+S') = \wh{f}(S)$ for $S'$ as defined in Claim~\ref{claim:useful-sum} and all $x,\gamma,S$.

\subsubsection{Representing the correlation in terms of the Fourier expansion}

In this subsection we present a sequence of lemmas that shall be used in our proof. These lemmas allow us to represent the correlation between different obfuscations in terms of the Fourier expansion, and will be helpful in bounding the correlation using discrete Fourier analysis in the following subsections.
Note that all inner products from now on are between elements of $\Z_p^r$, and consequently, their results lie in $\Z_p$.

\paragraph{An alternative representation of the obfuscation.}

In our proof, we shall frequently use the following alternative view of the obfuscation.

\begin{claim}\label{claim:dist}
		A pair $(x,y)$ taking a $(p,q,r)$-arithmetic-distribution may be sampled by drawing $x \sim \integers_p^r$ uniformly at random, choosing $v \sim \{ (1-q)/2, (3-q)/2, \ldots, (q-1)/2 \}^r$ uniformly,
		along with $\alpha \sim \integers_{p\cdot q}^{\ast}$
		and $\gamma, \gamma' \sim \integers_{p}^{\ast}$, and setting
		\[
			\forall i \in [r] \cc \quad y_i = (\gamma x_i + \gamma' \lfloor \alpha v_i / q \rceil) \modd p.
		\]
\end{claim}
\begin{proof}
		Assume $(x,y)$ has a $(p,q,r)$-arithmetic distribution, that is, $x_i = \gamma_1 \cdot \lfloor \alpha_1 \cdot \z_i / q \rceil \% p$ and $y_i = \gamma_2 \cdot \lfloor \alpha_2 \cdot \z_i / q \rceil \% p$, where $\z,\alpha_1,\alpha_2,\gamma_1,\gamma_2$ are as in Definition~\ref{def:noise2}.
		Write $\alpha_1 \z_i \% (pq) = q (\gamma_1^{-1} x_i \% p) + v_i$ with $v_i \in \{ (1-q)/2, \ldots, (q-1)/2 \}$. Note that, under the fixing of any (invertible) $\gamma_1, \alpha_1$, each pair of $(x_i, v_i) \in \integers_p \times \{ (1-q)/2, \ldots, (q-1)/2 \}$ arises from exactly one $\z_i \in \integers_{pq}$. Since $\z_i$ is uniformly distributed and independent of $\gamma_i, \alpha_i$, we see that $x_i, v_i$ are uniformly distributed (as stated in the claim) and independent of each other and of $\gamma_i, \alpha_i$.

		Finally,
		\[
			y_i = \gamma_2 \cdot \lfloor \alpha_2 \cdot \alpha_1^{-1} (q (\gamma_1^{-1} x_i \% p) + v_i) / q \rceil \% p,
		\]
		meaning that
		\[
			y_i = ((\gamma_2 \alpha_2 \alpha_1^{-1} \gamma_1^{-1} \% p) x_i + \gamma_2 \lfloor \alpha_2 \alpha_1^{-1} v_i / q \rceil) \% p.
		\]
		Denoting $\gamma \defeq (\gamma_2 \alpha_2 \alpha_1^{-1} \gamma_1^{-1} \% p)$, $\gamma' \defeq \gamma_2$ and $\alpha \defeq \alpha_2 \alpha_1^{-1}$, we have $y_i = (\gamma x_i + \gamma' \lfloor \alpha v_i / q \rceil) \modd p$. Note that $\gamma, \gamma', \alpha$ have the asserted distribution and are independent of $x, v$.
	\end{proof}

\paragraph{A quantity representing the contribution of $\wh{B}'(S)$ to the correlation.}
	We now formally introduce the notion $M_{p,q,r}(S)$ that will play a central role in the proof. The relevance of the notion to the correlation we study is shown in Lemma~\ref{lem:noise2} below.
	\begin{definition}\label{def:mag}
	 Let $S \in \integers_{p}^r$, and let $(x,y)$ be a pair that has a $(p,q,r)$-arithmetic distribution.
	 Define the magnitude of $S$ as (the real number)
	 \[
	 	M_{p,q,r}(S) \defeq \be[\chi_S(x) \overline{\chi_S(y)}].
	 \]
	\end{definition}

	\begin{lemma}[Orthogonality]\label{lem:mag-support}
		Let $S, S' \in \integers_{p}^r$, and assume $(x,y)$ has a $(p,q,r)$-arithmetic distribution.
		If $S' = \gamma' S$ with $\gamma' \in \integers_p^\ast$, then we have
		\begin{equation}\label{eq:orth1}
	 		\be[\chi_S(x) \overline{\chi_{S'}(y)}] = M_{p,q,r}(S).
	 	\end{equation}
	 	Otherwise (if $S' \neq \gamma' S$ for all $\gamma' \in \integers_p^\ast$),
	 	\begin{equation}\label{eq:orth2}
	 		\be[\chi_S(x) \overline{\chi_{S'}(y)}] = 0.
	 	\end{equation}
	\end{lemma}
	\begin{proof}
	To verify~\eqref{eq:orth1}, assume $S' = \gamma' S$ with $\gamma' \in \integers_p^\ast$.
	Note that by Definition~\ref{def:noise2}, if $(x,y)$ has a $(p,q,r)$-arithmetic-distribution, then $(x, \gamma'^{-1} y)$ admits a $(p,q,r)$-arithmetic-distribution as well. Hence,
	\[
		\be[\chi_S(x) \overline{\chi_{S'}(y)}] = \be[\chi_S(x) \overline{\chi_{S'}(\gamma'^{-1} y)}] = \be[\chi_S(x) \overline{\chi_{S}(y)}] = M_{p,q,r}(S).
	\]

	To verify~\eqref{eq:orth2}, recall that Claim~\ref{claim:dist} shows that $x \sim \integers_p^r$ and $y_i = \gamma x_i + \gamma' u_i$, where $\gamma, \gamma' \sim \integers_p^\ast$, and $u$ is independent of $x$ (its distribution is irrelevant for the current proof). Condition on the values of $\gamma, \gamma', u$, so that
	\[
		\be[\chi_S(x) \overline{\chi_{S'}(y)}] = \be_{\gamma, \gamma', u} [ \be_{x} \li[ \given{e_p ( \langle S, x \rangle - \langle S', y \rangle)}{\gamma, \gamma', u} \ri] ].
	\]
	Notice that given $\gamma, \gamma',$ and $u$, the expression $\langle S, x \rangle - \langle S', y \rangle$ is linear in $x$, and is \textbf{non-constant} (for all $\gamma, \gamma', u$), since $S, S'$ are non-proportional. Hence, $\be[e_p ( \langle S, x \rangle - \langle S', y \rangle)] = 0$, since $x \sim \integers_p^r$, and the expected value of $e_p(x')$ when $x' \sim \integers_p$ is uniformly distributed is $0$.
	\end{proof}

\newtheorem*{notation*}{Notation}
	\begin{notation*}
		Given two vectors $S, S' \in \integers_p^r$ and two scalars $\alpha, \beta \in \integers_p$, we denote by $\alpha S + \beta S' \in \Z_p^r$ the vector $S''$ having for all $i \in [r]$,
		\[
			S''_i = (\alpha S_i + \beta S'_i) \modd p.
		\]
	\end{notation*}
	
\paragraph{Representing the correlation in terms of $M_{p,q,r}$ and the Fourier expansion.}
	The following lemma shows how $M_{p,q,r}$ can be used to estimate the correlation $\cov(B'(x), B'(y))$ we aim at bounding, thus establishing~\eqref{eq:summary-bound1}.
	\begin{lemma}\label{lem:noise2}
		Let $\func{f}{\integers_{p}^r}{\complex}$ have $f(\gamma x) = f(x)$ for all $x \in \integers_{p}^r$, and $\gamma \in \integers_p^\ast$. Suppose $(x, y)$ has a $(p,q,r)$-arithmetic-distribution. Then
		\[
			\cov(f(x), f(y)) =
			\sum_{\substack{S \in \integers_{p}^r \\ S \neq 0_r}}
			(p-1) M_{p,q,r}(S) \wh{f}(S)^2.
		\]
	\end{lemma}
	\begin{proof}
		Using the expansion $f(x) = \sum_{S \in \integers_p^r} \wh{f}(S) \chi_S(x)$, we find
		\[
			\cov(f(x), f(y)) = \be_{x,y} [ (f(x) - \wh{f}(0)) \overline{(f(y) - \wh{f}(0))}] = \sum_{S,S' \in \integers_p^r \sm \{0_r\}} \wh{f}(S) \wh{f}(S') \be_{x,y}[\chi_S(x) \overline{\chi_{S'}(y)}].
		\]
		By comparing coefficients, the assumption that $f(\gamma x) = f(x)$ (for all $x$) implies $\wh{f}(S) = \wh{f}(\gamma S)$. Combining with Lemma~\ref{lem:mag-support} we get
		\[
			\cov(f(x), f(y)) = \sum_{S \neq 0_r} \sum_{\gamma \in \integers_p^\ast} \wh{f}(S) \wh{f}(\gamma S) M_{p,q,r}(S) = (p-1) \sum_{S \neq 0_r} \wh{f}(S)^2 M_{p,q,r}(S).
		\]
	\end{proof}

\paragraph{Bounding $M_{p,q,r}(S)$.} The following two lemmas allow us to bound $M_{p,q,r}(S)$.
	\begin{lemma}\label{lem:Mpqr-simp}
		Suppose $(x, y)$ has a $(p,q,r)$-arithmetic-distribution. For any $r \in \mathbb{N}$, any primes $p,q > 0$, and any non-zero $S \in \integers_p^r$, we have
		\begin{equation}\label{eq:mag3}
		M_{p,q,r}(S) =  \frac{p^2 \pr[\langle S, x\rangle = \langle S, y\rangle = 0 ] - 1}{(p-1)^2}.
		\end{equation}
	\end{lemma}
	\begin{proof}
		Assume $(x,y)$ has a $(p,q,r)$-arithmetic-distribution, and let $S \in \integers_p^r$ be any vector. We show
		\begin{equation}\label{eq:mag1}
		\begin{aligned}
			M_{p,q,r}(S) &=
			1 \cdot \pr[\langle S, x\rangle = \langle S, y\rangle = 0 ] \\
			&\phantom{=} - \frac{1}{p-1} \pr[\langle S, x\rangle = 0 \text{ XOR } \langle S, y\rangle = 0] \\
			&\phantom{=} + \frac{1}{(p-1)^2} \pr[\langle S, x\rangle \neq 0 \andd \langle S, y\rangle \neq 0], \\
		\end{aligned}
		\end{equation}
		where $(\langle S, x\rangle = 0 \text{ XOR } \langle S, y\rangle = 0)$ denotes the event that exactly one of $\langle S,x \rangle =0$ , $\langle S,y \rangle = 0$ holds.
		To verify~\eqref{eq:mag1}, note that if $(x,y)$ has a $(p,q,r)$-arithmetic distribution, then so does $(\gamma_1 x, \gamma_2 y)$, for any $\gamma_1,\gamma_2 \in \integers_p^*$. Hence,
		\[
		\begin{aligned}
			\be_{x,y}[ \chi_S(x) \overline{\chi_S(y)} ] &= \be_{x,y}[e_p ( \langle S , x\rangle - \langle S , y\rangle )] \\
			& = \be_{x,y}[e_p ( \gamma_1 \cdot \langle S , x\rangle - \gamma_2 \cdot \langle S , y\rangle )].
		\end{aligned}
		\]
		Letting $\gamma_1, \gamma_2$ be uniformly distributed in $\integers_p^{\ast}$ (independently of $(x,y)$), one can verify that for any fixed $x,y$ we have
		\begin{equation}\label{eq:mag2}
		\begin{gathered}
			\be_{\gamma_1, \gamma_2}[e_p ( \gamma_1 \cdot \langle S , x\rangle - \gamma_2 \cdot \langle S , y\rangle )] = \be_{\gamma_1}[e_p ( \gamma_1 \cdot \langle S , x\rangle )] \be_{\gamma_2} [e_p( \gamma_2 \cdot \langle S , y\rangle )]
			\\
			= \Big(\one\{\langle S , x\rangle = 0\} - \frac{\one\{\langle S , x\rangle \neq 0 \}}{p-1}\Big) \cdot
			\Big(\one\{\langle S , y\rangle = 0\} - \frac{\one\{\langle S , y\rangle \neq 0 \}}{p-1}\Big).
		\end{gathered}
		\end{equation}
		Equation~\eqref{eq:mag1} then follows by averaging~\eqref{eq:mag2} over $(x,y)$.

		\medskip Denote the three probabilities in~\eqref{eq:mag1} by $A,B,C$, that is:
		\[
		\begin{aligned}
			A &\defeq \pr[\langle S, x\rangle = \langle S, y\rangle = 0 ], \\
			B &\defeq \pr[\langle S, x\rangle = 0 \text{ XOR } \langle S, y\rangle = 0], \\
			C &\defeq \pr[\langle S, x\rangle \neq 0 \andd \langle S, y\rangle \neq 0].
		\end{aligned}
		\]
		Note that $A+B+C=1$, and that for any $S\neq 0$, we have $2A+B = 2\cdot \pr[\langle S,x\rangle = 0] = 2/p$, since each of $x,y$ is uniformly distributed in $\integers_p^r$. Substituting into~\eqref{eq:mag1} and simplifying, we obtain
		\[
			M_{p,q,r}(S) =  \frac{p^2 \pr[\langle S, x\rangle = \langle S, y\rangle = 0 ] - 1}{(p-1)^2},
		\]
as asserted.
	\end{proof}

	\begin{lemma}\label{lem:mag5}
		Let $r \in \mathbb{N}$, let $p \geq q > 0$ be prime numbers, and let $S \in \integers_p^r$ be a nonzero vector. If $(x,y)$ has a $(p,q,r)$-arithmetic-distribution, then
		\begin{equation}\label{eq:mag5}
			\pr_{x,y}[\langle S, x \rangle = \langle S, y \rangle = 0] \leq O(1 / (pq) + 1/p^2) \leq O(1/(pq)).
		\end{equation}
	\end{lemma}
	\begin{proof}
		Recall the distribution of $x,y$ given by Claim~\ref{claim:dist}, namely $y_i = (\gamma x_i + \gamma' u_i) \modd p$ with $u_i = \lfloor \alpha v_i / q \rceil \modd p$, where $x,v,\gamma, \gamma', \alpha$ are independent random variables. By the independence of $x, u$, we get
		\begin{equation}\label{eq:mag6.5}
			\pr_{x,y}[\langle S, x \rangle = \langle S, y \rangle = 0] =
			\pr_{x}[\langle S, x \rangle = 0] \pr_{u}[\langle S, u \rangle = 0].
		\end{equation}
		Since $x$ is uniformly random and $S$ is nonzero, $\pr_{x}[\langle S, x \rangle = 0] = 1/p$. It thus remains to show
		\begin{equation}\label{eq:offord1}
			\pr_{u}[\langle S, u \rangle = 0] \leq O(1 / q + 1 / p).
		\end{equation}
		To prove~\eqref{eq:offord1}, we note that it is a Littlewood-Offord-type statement: the entries $u_i$ are independent random variables, and $\langle S, u \rangle$ is their weighted sum, and we are concerned with the probability it attains a specific value.
	
		We tackle the problem by using a standard antichain argument. Roughly, choose an $i$ with $S_i \neq 0$, and condition on the values of $u_j$ for all $j \neq i$. Then, there is at most one value of $u_i$ that would make $\langle S, u \rangle = S_i u_i + \sum_{j: j \neq i} S_j u_j \modd p = 0$. Recall $u_i = \lfloor \alpha v_i / q \rceil \modd p$, where $v_i$ is uniformly distributed in $\{ (1-q)/2, \ldots, (q-1)/2 \}$. Hence, assuming that the map $v_i \mapsto u_i$ is injective, we have that $\langle S, u \rangle = 0$ with probability $\leq 1/q$, as required. This last assumption is not strictly correct, however it can be corrected as follows.

		Depending on $\alpha$ and on $\{u_j\cc j \neq i\}$, we let $Z = \{ \tau \in \{ (1-q)/2, \ldots, (q-1)/2 \} \cc (S_i \lfloor \alpha \tau / q \rceil  + \sum_{j: j \neq i} S_j u_j) \modd p = 0 \}$, and note that according to the above discussion, $\pr[\langle S,u \rangle = 0] = \be[|Z|]/q$. Hence,~\eqref{eq:offord1} is reduced to showing that $\be[|Z|] \leq O(1 + q/p)$. Write $Z = \{\tau_1, \ldots, \tau_k\}$ with $\tau_1 < \ldots < \tau_k$, and $Z' = \{\tau_2 - \tau_1, \tau_3 - \tau_1, \ldots, \tau_k - \tau_1\}$. Note that $|Z| = |Z'| + 1$, and that every $\sigma \in Z'$ satisfies $\alpha \sigma \in pq\integers + (-q,q)$ (i.e., the residue $(\alpha \sigma) \modd (pq)$ is either smaller than $q$, or larger than $pq-q$). Hence,
		\[
			\be[|Z|] \leq 1 + \be[|Z'|] \leq 1 + \be_{\alpha}[ \sum_{\sigma = 1}^{q-1} \one\{ \alpha \sigma \in pq\integers + (-q,q) \} ] = O(1 + q/p),
		\]
		where the last bound follows since for any fixed $\sigma$, $\alpha\sigma \modd (pq)$ is uniformly distributed in $\integers_{pq}^{\ast}$ (recall $\sigma < q\leq p$), and the probability it is in $pq\integers + (-q,q)$ is $O(1/p)$.

	\end{proof}

	Lemmas~\ref{lem:Mpqr-simp} and~\ref{lem:mag5} yield the following corollary, that upper bounds $M_{p,q,r}(S)$.

	\begin{corollary}\label{cor:mag-bound}
		Let $r \in \mathbb{N}$, let $p \geq q > 0$ be prime numbers, and let $S \in \integers_p^r$ be a nonzero vector. Then
		\[
			|M_{p,q,r}(S)| \leq O(1/(pq)).
		\]
	\end{corollary}
	\begin{proof}
		By~\eqref{eq:mag3} we have
		\[
			M_{p,q,r}(S) = \frac{p^2 \pr[\langle S, x\rangle = \langle S, y\rangle = 0 ] - 1}{(p-1)^2}.
		\]
		By~\eqref{eq:mag5} we have
		\[
			\pr[\langle S, x\rangle = \langle S, y\rangle = 0 ] \leq O(1/(pq)).
		\]
		These two estimates yield the desirable
		$
			|M_{p,q,r}(S)| = O(1/(pq)+1/p^2) = O(1/(pq)).
		$
	\end{proof}

\subsubsection{Partitioning into 2-dimensional subspaces}

While Corollary~\ref{cor:mag-bound}, in conjunction with Lemma~\ref{lem:noise2}, allows us bounding $\cov(B'(x), B'(y))$ from above (which is the main task we are tackling), the obtained upper bound is not sufficiently tight for our purposes. To achieve a stronger bound, we use the special structure of $B'$ observed in Claim~\ref{claim:useful-sum} -- namely, that it satisfies $B'(x)=B'(\gamma x)$ for all $\gamma \in \integers_p^\ast$, and that its Fourier expansion satisfies $\wh{B'}(S) = \wh{B'}(S+S')$ for all $S$ and the specific $S'$ defined in Claim~\ref{claim:useful-sum} -- to show that it is sufficient to bound the sums $\sum_{S \in U} M_{p,q,r}(S)$ over certain 2-dimensional subspaces $U$.

	\begin{lemma}\label{lem:non-cartesian2}
	Let $r \in \mathbb{N}$, and let $p \geq q$ be prime numbers.
	Let $\func{f}{\integers_p^r}{\czo}$ satisfy $f(x) = f(\gamma x)$ for all $x \in \integers_p^r$ and $\gamma \in \integers_p^\ast$. Furthermore, assume there exists a particular nonzero vector $S' \in \integers_p^r$ such that $\wh{f}(S+S') = \wh{f}(S)$ for all $S \in \integers_p^r$.
	
	Let $C$ be a constant such that for all $S$ satisfying $\forall \nu \in \integers_p \cc S \neq \nu S'$,
	\[
		\sum_{\substack{\eta \in \integers_{p}^{\ast} \\ \nu \in \integers_p}} M_{p,q,r}(\eta S + \nu S') \leq C.
	\]
	If $\mu = \be[f]$, and $(x,y)$ has a $(p,q,r)$-arithmetic-distribution, then
	\[
		\cov(f(x), f(y)) \leq O\Big(\frac{p}{q} \mu^2 +  \frac{C}{p} \mu \Big).
	\]
	\end{lemma}
	\begin{proof}
	By Lemma~\ref{lem:noise2}, we have
	\[
		\cov(f(x), f(y)) = (p-1) \sum_{S \neq 0_r} M_{p,q,r}(S) \wh{f}(S)^2
	\]
	There are two kinds of contributions to the right hand side, corresponding to elements $S$ with $S = \nu S'$, and to other elements.

	\noindent
	\textbf{Case 1:}
	The contribution of each $S$ with $S = \nu S'$ for $\nu \in \integers_p^\ast$, is $M_{p,q,r}(S') \cdot \wh{f}(S')^2$. (Recall that by Lemma~\ref{lem:mag-support}, for such an $S$, we have $M_{p,q,r}(S')=M_{p,q,r}(S)$, and since $f(x)=f(\gamma x)$, we have $\wh{f}(S)=\wh{f}(S')$.) We note that $|\wh{f}(S)| = |\be_x[f(x) \overline{\chi_S(x)}]| \leq \be |f(x)| = \mu$. Using the bound on $M_{p,q,r}(S)$ from Corollary~\ref{cor:mag-bound}, and the fact that there are only $p-1$ such $S$'s, we get that the total contribution in this case is
	\[
		(p-1)\sum_{S=\nu S'}M_{p,q,r}(S) \wh{f}(S)^2 \leq (p-1)^2 M_{p,q,r}(S') \mu^2 \leq p^2 O( 1/(pq)) \mu^2 \leq O\Big(\frac{p}{q} \mu^2 \Big).
	\]

	\noindent
	\textbf{Case 2:} The contribution of elements $S$ with $S \neq \nu S'$ for all $\nu \in \integers_p^\ast$ (we denote this family of elements by $\mathcal{S}$), can be analyzed using the assumption that $\wh{f}(S+ \nu S') = \wh{f}(S)$ for all $\nu \in \integers_p$. It follows that
	\[
		(p-1)\sum_{S \in \mathcal{S}} M_{p,q,r}(S) \wh{f}(S)^2 = \frac{p-1}{p(p-1)} \sum_{S \in \mathcal{S}} \wh{f}(S)^2 \cdot \sum_{\substack{\eta \in \integers_p^\ast \\ \nu \in \integers_p}} M_{p,q,r}(\eta S + \nu S') \leq \frac{C \mu}{p}.
	\]
	Here, the first equality holds since each summand on the left hand side appears $p(p-1)$ times on the right hand side. The final inequality is obtained by the assumption regarding $C$ and the estimate $\sum_S \wh{f}(S)^2 = \be[f^2] \leq \mu$.

	\noindent
	\textbf{Overall:} Combining the above two contributions we get
	\[
		\cov(f(x), f(y)) \leq O\Big( \frac{p}{q} \Big) \mu^2 + \frac{C}{p}\mu,
	\]
	as asserted.
	\end{proof}

\subsubsection{Bounding \tops{$\sum_S M_{p,q,r}(S)$} over 2-dimensional subspaces}
\label{sec:sub:lemma-proof}

	In this subsection we present the most complex step of the proof -- bounding $\sum_{S \in \mathcal{U}} M_{p,q,r}(S)$ over 2-dimensional subspaces $\mathcal{U}$, which will allow us to complete the proof in conjunction with Lemma~\ref{lem:non-cartesian2}. The proof is quite technical.
	Its core element is the representation of a sub-problem as a Littlewood-Offord type problem and the use of an antichain technique for handling it.
	\begin{lemma}\label{lem:mag-exp}
		Let $S, S' \in \integers_{p}^r$ be nonzero vectors with $S \neq \nu S'$ for all $\nu \in \integers_p^\ast$.
		Suppose $q \leq p$, then
		\begin{equation}\label{eq:mag-exp}
			\Big| \sum_{\substack{\mu \in \integers_{p}^{\ast} \\ \nu \in \integers_p}} M_{p,q,r}( \mu S + \nu S' ) \Big|  \leq O(p/q^2 + \log(q)/q + 1/p) \leq O(p\log(q) / q^2).
		\end{equation}
	\end{lemma}

\begin{proof}\skipline

		\paragraph{Step 1.}
		We express the left hand side of~\eqref{eq:mag-exp}, using Lemma~\ref{lem:Mpqr-simp} (with $T = \mu S + \nu S'$), as
		\begin{equation}\label{eq:tt1}
			D\defeq \sum_{\substack{\mu \in \integers_{p}^{\ast} \\ \nu \in \integers_p}} M_{p,q,r}( \mu S + \nu S' ) =
			\frac{p^2}{(p-1)^2} \sum_{\substack{\mu \in \integers_{p}^{\ast} \\ \nu \in \integers_p}} \pr[\langle \mu S + \nu S', x\rangle = \langle \mu S + \nu S', y\rangle = 0 ] - \frac{p}{p-1},
		\end{equation}
		where the probability is taken over pairs $(x, y)$ distributed according to a $(p,q,r)$-arithmetic-distribution.
		Adding and subtracting all pairs of the form $(\mu, \nu) = (0, \nu)$ to the sum in~\eqref{eq:tt1}, we get
        \[
			\frac{(p-1)^2}{p^2}D
			=
			\sum_{\substack{\mu, \nu \in \integers_p \\ (\mu, \nu) \neq (0,0)}} \pr[\langle \mu S + \nu S', x\rangle = \langle \mu S + \nu S', y\rangle = 0 ]
			- (p-1)\pr_{x,y}[\langle S', x \rangle = \langle S', y \rangle = 0]
			- \frac{p-1}{p}.
		\]
        Note that given $x,y$, the number of solutions $(\mu,\nu)$ of the equation system $(\langle\mu S + \nu S',x \rangle = 0) \wedge (\langle\mu S + \nu S',y \rangle = 0)$, is equal to the number of solutions of the system $(\mu \langle S,x \rangle + \nu \langle S,y \rangle = 0) \wedge (\mu \langle S',x \rangle + \nu \langle S',y \rangle = 0)$. Indeed, they are equal to the sizes of the left kernel and the right kernel of the matrix
        \[
        	\begin{pmatrix}
           		\langle S, x \rangle & \langle S, y \rangle \\
           		\langle S', x \rangle & \langle S', y \rangle \\
         	\end{pmatrix},
		\]
        which are known to be equal.
        The latter linear system may succinctly be written as $\mu V(x)+\nu V(y)=0$ where
		\[
			V(x) \defeq
			\begin{pmatrix}
           		\langle S, x \rangle \\
           		\langle S', x \rangle \\
         	\end{pmatrix} \in \integers_p^2,
         	\qquad
         	V(y) \defeq
			\begin{pmatrix}
           		\langle S, y \rangle \\
           		\langle S', y \rangle \\
         	\end{pmatrix} \in \integers_p^2.
		\]
        Using this equality, we obtain
		\[
			\frac{(p-1)^2}{p^2}D
			=
			\sum_{\substack{\mu, \nu \in \integers_p \\ (\mu, \nu) \neq (0,0)}} \pr_{x,y}[ \mu V(x) + \nu V(y) = 0 ]
			- (p-1)\pr_{x,y}[\langle S', x \rangle = \langle S', y \rangle = 0]
			- \frac{p-1}{p}.
		\]
		Note that~\eqref{eq:mag-exp} is equivalent to $D \leq O(p/q^2 + \log(q)/q + 1/p)$, which follows from the bounds:
		\begin{equation}\label{eq:mag55}
			\pr_{x,y}\big[\langle S', x \rangle = \langle S', y \rangle = 0 \big] \leq O(1 / (pq) + 1/p^2),
		\end{equation}

		\begin{equation}\label{eq:mag6}
			\Bigg| \sum_{\substack{\mu, \nu \in \integers_p \\ (\mu, \nu) \neq (0,0)}} \pr_{x,y}[ \mu V(x) + \nu V(y) = 0 ] - \frac{p-1}{p} \Bigg| \leq O(p/q^2 + \log(q) / q + 1/p).
		\end{equation}
		As~\eqref{eq:mag55} follows from Lemma~\ref{lem:mag5}, we are left with the task of verifying~\eqref{eq:mag6}.

		\paragraph{Step 2.}
		We verify~\eqref{eq:mag6}. Note that $\frac{p^2-1}{p^2} - \frac{p-1}{p} < 1/p$, thus~\eqref{eq:mag6} follows from
		\begin{equation}\label{eq:mag7}
			\Bigg| \sum_{(\mu, \nu) \neq (0,0)} \pr_{x,y}[ \mu V(x) + \nu V(y) = 0 ] - \frac{p^2-1}{p^2} \Bigg| \leq O(p/q^2 + \log(q) / q + 1/p).
		\end{equation}
		In case of either $\mu = 0$ or $\nu = 0$, we have $\pr_{x,y}[ \mu V(x) + \nu V(y) = 0 ] = 1/p^2$, since both $x,y$ are uniformly distributed in $\integers_p^r$ and $S, S'$ are linearly independent vectors. We must hence verify
		\begin{equation}\label{eq:mag8}
			\Bigg| \sum_{\mu \in \integers_p^{\ast}} \pr_{x,y}[ \mu V(x) = V(y) ] - \frac{p-1}{p^2} \Bigg| \leq O(1/q^2 + \log(q) / (pq) + 1/p^2).
		\end{equation}
		We now reason about the left hand side of~\eqref{eq:mag8}. Specifically, we consider the three sums
		\[
			Q_1 \defeq \sum_{\mu \in \integers_p^{\ast}} \pr_{x,y} \Big[ \mu V(x) = V(y) \,\andd\, \Big(\{\langle S, x \rangle, \langle S, y \rangle\} = \{0\}\Big) \Big],
		\]
		
		\[
			Q_2 \defeq \sum_{\mu \in \integers_p^{\ast}} \pr_{x,y} \Big[ \mu V(x) = V(y) \,\andd\, \Big( 0 \in \{\langle S, x \rangle, \langle S, y \rangle\} \neq \{0\} \Big) \Big],
		\]

		\[
			Q_3 \defeq \sum_{\mu \in \integers_p^{\ast}} \pr_{x,y} \Big[ \mu V(x) = V(y) \,\andd\, \Big( 0 \notin \{\langle S, x \rangle, \langle S, y \rangle\} \Big) \Big],
		\]
		and show the following estimates, which together imply~\eqref{eq:mag8}:
		\begin{equation}\label{eq:mag9}
			|Q_1| \leq O(1/(pq) + 1/p^2),
		\end{equation}
		\begin{equation}\label{eq:magA}
			Q_2 = 0,
		\end{equation}
		\begin{equation}\label{eq:magB}
			\Big|Q_3 - \frac{p-1}{p^2} \Big| \leq O(1/q^2 + \log(q) / (pq) + 1/p^2).
		\end{equation}

		In order to obtain~\eqref{eq:magA}, notice that if one of $\langle S, x\rangle, \langle S, y\rangle$ is zero, and the other is not, then there cannot be a $\mu \in \integers_p^{\ast}$ which is the quotient of them.

		\paragraph{Step 3.} We prove~\eqref{eq:mag9}. First, we observe
		\[
			Q_1 \leq \pr_{x,y}\big[\{\langle S, x \rangle, \langle S, y \rangle\} = \{0\} \big] + (p-2)\pr_{x,y}\big[V(x)=V(y)=0\big].
		\]
		This is because upon fixing $x,y$, whenever $\langle S', x \rangle \neq 0$ or $\langle S', y \rangle \neq 0$, there is at most one value of $\mu \in \integers_{p}^\ast$ for which $\mu V(x) = V(y)$. Lemma~\ref{lem:mag5} implies
		\[
		\pr_{x,y}\big[\{\langle S, x \rangle, \langle S, y \rangle\} = \{0\} \big] \leq O(1/(pq) + 1/p^2).
		\]
		Hence, it remains to show
		\begin{equation}\label{eq:magC}
			\pr_{x,y}\big[V(x)=V(y)=0\big] \leq O(1/(p^2q) + 1/p^3).
		\end{equation}
		Since $x\sim \integers_{p}^r$ is uniformly distributed, and $S, S'$ are two independent vectors, then we have $\pr[V(x) = 0] = 1/p^2$. Moreover, similarly to the reasoning in Lemma~\ref{lem:mag5} (\eqref{eq:mag6.5} in particular),
		\[
			\pr\li[ \given{V(y) = 0}{V(x)=0} \ri] = \Pr[\langle S, u\rangle = \langle S', u\rangle = 0],
		\]
		where $u = \lfloor \alpha v_i / q \rceil$ with $\alpha \sim \integers_{pq}^\ast$ and $v_i \sim \{ (1-q)/2, \ldots, (q-1)/2 \}$ are uniformly distributed. But, according to~\eqref{eq:offord1}, we have
		\[
			\Pr[\langle S, u\rangle = \langle S', u\rangle = 0] \leq \Pr[\langle S, u\rangle = 0] \leq O(1/p+1/q),
		\]
		which, together with $\Pr[V(x)=0]=1/p^2$ implies~\eqref{eq:magC}.

		\paragraph{Step 4.} We prove~\eqref{eq:magB}. Recall, again, that $x,y$ are sampled by taking $x\sim \integers_p^r$ uniformly at random, and setting $y_i = (\gamma x_i + \gamma' u_i) \modd p$ with $\gamma', \gamma \sim \integers_p^\ast$ and $u_i = \lfloor \alpha \cdot v_i \rceil$, with $\alpha \sim \integers_{pq}^\ast$ and $v_i \sim \{ (1-q)/2, \ldots, (q-1)/2 \}$. We further decompose $Q_3$ into two parts, $Q_3 = Q_4 + Q_5$, according to whether $\gamma \langle S , x \rangle = \langle S, y \rangle$ or not:
		\[
			Q_4 = \sum_{\mu \in \integers_p^{\ast}} \pr_{x,y} \Big[ \mu V(x) = V(y) \,\andd\, \mu \neq \gamma \,\andd\, 0 \notin \{\langle S, x \rangle, \langle S, y \rangle\} \Big],
		\]
		\[
			Q_5 = \pr_{x,y} \Big[ \gamma V(x) = V(y) \,\andd\, 0 \notin \{\langle S, x \rangle, \langle S, y \rangle\} \Big],
		\]
		and define the auxiliary probabilities
		\[
		\begin{aligned}
			\beta &= \Pr\big[\gamma \langle S, x \rangle = \langle S, y \rangle \,\andd\, \big(0 \notin \{\langle S, x \rangle, \langle S, y \rangle\} \big) \big],
			\\
			\eta &= \Pr[0 \notin \{\langle S, x \rangle, \langle S, y \rangle\} ].
		\end{aligned}
		\]

		\noindent We make four claims:
		\begin{itemize}
			\item \textbf{Simplifying $Q_4$:} \qquad $Q_4 = (\eta - \beta)/p$,
			\item \textbf{Upper bounding $\beta$:} \qquad $\beta \leq O(1/p + 1/q)$,
			\item \textbf{Lower bounding $\eta$:} \qquad $\eta \geq 1-2/p$,
			\item \textbf{Upper bounding $Q_5$:} \qquad $Q_5 \leq O(1/q^2 + \log(q)/(pq))$.
		\end{itemize}
		Since $Q_3 = Q_4 + Q_5$, these claims clearly imply~\eqref{eq:magB}.

		\paragraph{Lower bounding $\eta$.} Recall that both $x$ and $y$ are uniformly distributed in $\integers_p^r$, and so $\langle S, x \rangle$ and $\langle S, y \rangle$ are uniformly distributed in $\integers_p$. Thus, by a union bound, we have $\eta \geq 1-2/p$, as asserted.

		\paragraph{Upper bounding $\beta$.}
		Recall that $y_i = (\gamma x_i + \gamma' u_i) \modd p$, and hence the event that $\gamma \langle S, x \rangle = \langle S, y \rangle$ is exactly the event that $\langle S, u \rangle = 0$. The probability of this latter event may be upper bounded by $O(1/p+1/q)$ using~\eqref{eq:offord1}. Hence,
		\[
			\beta \leq \Pr[\langle S, u \rangle = 0] \leq O(1/p+1/q),
		\]
		as asserted.

		\paragraph{Simplifying $Q_4$.} As an appetizer, note that if we would replace in $Q_4$ the requirement of $\mu V(x) = V(y)$ by $\mu \langle S, x\rangle = \langle S, y\rangle$, and call the result $Q_4'$, then we would get $Q_4' + \beta = \eta$. All that is left in order to prove~(1) is to show that $p Q_4 = Q_4'$.
			
		Observe that we may assume that
		\begin{equation}\label{Eq:Assumption0}
		\exists \ell \in [r]: S_\ell = 0 \andd S'_\ell \neq 0.
		\end{equation}
		To reduce to this case, we choose any $\ell$ with $S'_\ell \neq 0$, and replace $S$ by $S - \frac{S_\ell}{S'_\ell} S'$, which is also nonzero. (Note that the sum on the left hand side of~\eqref{eq:mag-exp} does not change by this replacement.)
		
		We condition on the values of $\gamma, u, \gamma'$ and $\{x_j : j \neq \ell\}$ (i.e. on the $\sigma$-algebra generated by these variables). The only information that is missing in the probability space is $x_\ell$ -- it is uniformly distributed under the current conditioning. While the contribution to $Q_4'$ is fixed under the current conditioning (as we assumed $S_\ell=0$), we claim there exists exactly one value of $x_\ell$ that would contribute to the probability expressed by $Q_4$. To see this, let $\mu$ be the unique element of $\integers_p^\ast$ that has $\mu \langle S, x \rangle = \langle S, y \rangle$. In order to have $\mu V(x) = V(y)$ we must also have $\mu \langle S', x \rangle = \langle S', y \rangle$. Under our conditioning, this latter equation is a linear equation in $x_\ell$ with the linear coefficient $(\mu - \gamma)S'_\ell$ (recall how $y_\ell$ depends on $x_\ell$), and some constant coefficient which is deterministic under our conditioning. This equation has a unique solution in $x_\ell$. Since $x_\ell$ has a uniform distribution, we have $Q_4 = Q_4' / p = (\eta-\beta)/ p$, as asserted.

		\paragraph{Step 5.} Lastly, we upper bound $Q_5$. Recall
		\[
			V(y) =
			\begin{pmatrix}
           		\langle S, y \rangle \\
           		\langle S', y \rangle \\
         	\end{pmatrix}
         	=
         	\gamma \begin{pmatrix}
           		\langle S, x \rangle \\
           		\langle S', x \rangle \\
         	\end{pmatrix}
         	+ \gamma'
         	\begin{pmatrix}
           		\langle S, u \rangle \\
           		\langle S', u \rangle \\
         	\end{pmatrix}.
		\]
		Hence, the event $\gamma V(x) = V(y)$ is simply $\{\langle S, u \rangle = 0 \, \andd \, \langle S', u \rangle = 0\}$, implying
		\[
			Q_5 \leq \Pr_u[ \langle S, u \rangle = 0 \, \andd \, \langle S', u \rangle = 0 ].
		\]
		We use an argument similar to the argument we had in Lemma~\ref{lem:mag5} (specifically,~\eqref{eq:offord1}) to upper bound this last quantity.
		Let $l \in [r]$ be any coordinate with $S_l \neq 0$, and let $\ell$ be a coordinate with $S'_\ell \neq 0$ while $S_\ell = 0$, whose existence we assumed (see~\eqref{Eq:Assumption0}). Recall $u_j = \lfloor \alpha v_j/ q\rceil$. Condition on any specific values for $\alpha, \{u_j : j \notin \{l,\ell\}\}$. Similarly to Step $4$, we can upper bound the probability that $\langle S, u \rangle = 0$ by the probability that $u_l$ turns out to be just the right value that would make $\langle S, u \rangle = 0$ true. This probability is upper bounded by $\frac{1}{q}$ times the maximal number of elements $v'' \in \{ (1-q)/2, \ldots, (q-1)/2 \}$ that map to the same $u''$ under $u'' = \lfloor \alpha v'' / q\rceil \modd p$. Using the same reasoning as in Lemma~\ref{lem:mag5}, we see that this probability is upper bounded by $(|Z_\alpha|+1) / q$, where
		\[
			Z_\alpha = \{ \sigma \in \{ (1-q)/2, \ldots, (q-1)/2 \} : (\alpha \sigma) \in pq \integers + (-q, q) \}.
		\]
		Under the described conditioning, $\langle S, u \rangle$ is determined by $u_l$, and by further conditioning on the value of $u_l$, $\langle S', u \rangle$ is determined by $u_\ell$. Thus, the event that both these quantities are equal zero has probability
		\begin{equation}\label{eq:second-moment}
		\Pr[\langle S, u\rangle = 0 \andd \langle S', u\rangle = 0] \leq (|Z_\alpha| + 1)^2 / q^2.
		\end{equation}
		Our task of upper bounding $Q_5$ is hence reduced to understanding the second moment of $|Z_\alpha|$.
		For any fixed $\alpha$,
		\[
		\begin{aligned}
			|Z_\alpha|^2
			&= \sum_{\sigma_1=(1-q)/2}^{(q-1)/2} \sum_{\sigma_2 = (1-q)/2}^{(q-1)/2} \one\big\{ \{\sigma_1 \cdot \alpha, \sigma_2 \cdot \alpha\} \seq pq\integers + (-q, q) \big\} \\
			&\leq 4\sum_{\sigma_1=0}^{q-1} \sum_{\sigma_2 = 0}^{q-1} \one\big\{ \{\sigma_1 \cdot \alpha, \sigma_2 \cdot \alpha\} \seq pq\integers + (-q, q) \big\}.
		\end{aligned}
		\]
		Denote $\tau_j' \defeq (\sigma_j \cdot \alpha \modd pq)$, then we have $\tau_1' \cdot \sigma_2 \modd (pq) = \tau_2' \cdot \sigma_1  \modd (pq)$. Using the assumption $q \leq p$, we may leverage this equation into an equation over the integers (i.e. not involving a modulus), in the following way.

Let $\tau_j \defeq \min\{\tau_j', pq - \tau_j'\}$, so the above equation reads either $\tau_1 \cdot \sigma_2 \modd (pq) = \tau_2 \cdot \sigma_1  \modd (pq)$ or $\tau_1 \cdot \sigma_2 \modd (pq) = -\tau_2 \cdot \sigma_1  \modd (pq)$. Observe that $\tau_j \leq q$, and hence, $\tau_1 \sigma_2, \tau_2 \sigma_1 \leq q^2 \leq pq$.  Therefore, over the integers we must have
\begin{equation}\label{Eq:AAux1}
\tau_1 \cdot \sigma_2 = \tau_2 \cdot \sigma_1 \qquad \mbox{ or } \qquad \tau_1 \cdot \sigma_2 + \tau_2 \cdot \sigma_1 = pq.
\end{equation}
In order to bound from above the expectation of $|Z_\alpha|^2$ over the $|\integers_{pq}^\ast|=\phi(pq)$ possible values of $\alpha$, we note that $\alpha$ can be recovered uniquely from $\sigma_1,\tau'_1$, as $\alpha= \tau'_1 \cdot \sigma_1^{-1} \modd pq$. Thus, any given quadruple $(\sigma_1,\sigma_2,\tau_1,\tau_2)$ corresponds to at most two different values of $\alpha$, and so, when we sum over all values of $\alpha$, each solution of each of the equations in~\eqref{Eq:AAux1} is counted at most twice.     		

Therefore, we have
		\begin{equation}\label{eq:magD}
			\be_{\alpha}[|Z_\alpha|^2] \leq \frac{4 \cdot 2}{\phi(pq)} \Big( \# \{(a,b,c,d) : ab = cd \} + \# \{(a,b,c,d) : ab + cd = pq \} \Big),
		\end{equation}
		where $a,b,c,d$ take values in $\{0,1,\ldots, q-1\}$.
		
		We bound the number of such quadruples $(a,b,c,d)$ by the following simple number-theoretic lemma, whose proof is given below.
		
		\begin{lemma}\label{lem:number-theory}
			Let $q,N > 0$ be positive integers. Define
			\[
			\begin{aligned}
			P &= \set{(a, b, c, d)}{ab+cd = N} \seq \{0,1,\ldots, q-1\}^4, \\
			Q &= \set{(a, b, c, d)}{ab=cd} \seq \{0,1,\ldots, q-1\}^4.
			\end{aligned}
			\]
			Then, $|P| \leq O(q^2 \log(q))$ and $|Q| \leq O(q^2 \log(q))$.
		\end{lemma}
		By Lemma~\ref{lem:number-theory}, the number of these quadruples $(a,b,c,d)$ is $O(q^2 \log(q))$. Combining~\eqref{eq:second-moment} and~\eqref{eq:magD} with Lemma~\ref{lem:number-theory}, we arrive at
		\[
			Q_5 \leq O\Big( \frac{1}{q^2} \cdot \big( 1 + \frac{q^2 \log(q)}{pq} \big) \Big) \leq
			O\Big( \frac{1}{q^2} + \frac{\log(q)}{pq} \Big),
		\]
		concluding the proof of Lemma~\ref{lem:mag-exp}.

		\paragraph{Summary of the proof of Lemma~\ref{lem:mag-exp}.} We parsed the left hand side of~\eqref{eq:mag-exp}, and interpreted it as the bias introduced in an event, ($V(x)$ is proportional to $V(y)$) caused by dependence between $x,y\sim \integers_p^r$. The core of the argument upper bounds this bias by posing the problem as a Littlewood-Offord-type problem, and using an antichain argument along with simple number theoretic estimates.
\end{proof}

\begin{proof}[Proof of Lemma~\ref{lem:number-theory}]
	First, notice that we may consider, in both cases, $a,b,c,d > 0$, as there are only $O(q^2)$ quadruples with $0 \in \{a,b,c,d\}$ and either $ab=cd$ or $ab+cd = N$. Indeed, regarding $ab=cd$, we must have $0$ on both sides, which implies that there are only $O(q^2)$ possible pairs. Regarding $ab+cd=N$, if $a=0$, then $b,c$ have $q^2$ options, and they determine $d$, totaling in $\leq O(q^2)$ pairs. We call the analogs of $P,Q$, with the quadruples containing $0$ removed, $P', Q'$, respectively.
	
	Second, we count $|P'|$. Fixing $a,c$, we see that $b,d$ must satisfy the linear equation $ab + cd = N$. The different solutions $(b,d)$ for this equation differ by integral multiples of the vector $(c/\gcd(a,c), -a/\gcd(a,c))$. Since both $b,d$ are integers in $[1,q)$, the number of such solutions is at most $q \gcd(a,c) / \max(a,c)$. Denoting $g = \gcd(a,c)$ we arrive at

	\[
	|P'| \leq \sum_{g=1}^{q} \sum_{g | a} \sum_{g|c} \frac{q g}{\max(a,c)} \leq \sum_{g=1}^{q} \sum_{g|c} \frac{2c}{g} \cdot \frac{q g}{c} \leq \sum_{g=1}^{q} \sum_{g| c} 2q \leq \sum_{g=1}^{q} 2q^2/g = O(q^2 \log(q)),
	\]
	as required.
	
	Lastly, bounding $|Q'|$ is done likewise, this time considering the equation $ab-cd = 0$.

\end{proof}
We note that the  $\log(q)$ factor in the bound $O\Big( \frac{1}{q^2} + \frac{\log(q)}{pq} \Big)$ of Lemma~\ref{lem:number-theory} is the reason for the logarithmic loss in Theorem~\ref{thm:red-sum}.
Unfortunately, one can show that the assertion of Lemma~\ref{lem:number-theory} is tight up to a constant factor, at least regarding the size of $Q$.

\subsubsection{Wrapping up the proof of the obfuscation lemma}

	\begin{proof}[Proof of Lemma~\ref{lem:ss-bound}]
		Recall that we assume an algorithm $\func{B}{\integers_p^r}{\binom{[r]}{k}}$ which always reports a $k$-tuple of its input numbers whose sum is $0$ modulo $p$ (and is allowed to report failure). Using $B$, we define the obfuscation algorithm $\func{A}{\integers_{pq}^r}{\binom{[r]}{k}}$ (Algorithm~\ref{alg:sum}) which reports a $k$-tuple of its input numbers whose sum is $0$ modulo $pq$.
We further define $\func{B'}{\integers_p^r}{[0,1]}$ by
		\[
			B'(x) = \Pr_{P,\gamma}[B(P(\gamma \cdot x)) = P(\{1,2,\ldots, k\})],
		\]
		where $P \sim S_r$ and $\gamma \sim \integers_p^\ast$, and the probability is taken also over $B$'s internal randomness. By Lemma~\ref{lem:to-bounded-sum}, we have
        \[
			\sum_{i = 1}^{k} x_i \neq 0 \quad \implies \quad B'(x) = 0,
		\]
		and
		\begin{equation}\label{eq:ss1}
			\mu \defeq \be[B'(x)] \leq 1/\binom{r}{k},
		\end{equation}
		and if $(x,y)$ has a $(p,q,r)$-arithmetic-distribution then
		\begin{equation}\label{eq:ss2}
			\pr[P^{-1}(B(P(x))) = Q^{-1}(B(Q(y)))] = \binom{r}{k} \be_{x,y} [B'(x) B'(y)].
		\end{equation}
		Furthermore, by Claim~\ref{claim:useful-sum}, $B'(x) = B'(\gamma x)$ for all $\gamma \in \integers_p^\ast$
and $\wh{B'}(S+S') = \wh{B'}(S)$ for all $S$, with $S'$ as defined in Claim~\ref{claim:useful-sum}.
		
Hence, we may apply Lemma~\ref{lem:non-cartesian2} with
		\[
		C = O(p\log(q) /q^2),
		\]
		as provided by Lemma~\ref{lem:mag-exp} (notice that we assume $p \geq q$), to conclude that
		\[
			\cov(B'(x), B'(y)) \leq O\Big( \frac{p}{q} \mu^2 + \frac{C}{p}\mu \Big) = O\Big( \frac{p}{q} \mu^2 + \frac{\log(q)}{q^2} \mu \Big).
		\]
		Notice that as both $x,y$ are uniformly distributed in $\integers_p^r$, we have
		\[
		\be[B'(x) B'(y)] = \mu^2 + \cov(B'(x), B'(y)).
		\]
		Combining these estimates with~\eqref{eq:ss1} and~\eqref{eq:ss2}, we obtain
		\[
			\pr[P^{-1}(B(P(x))) = Q^{-1}(B(Q(x)))] \leq O\Big( \frac{p}{q\binom{r}{k}} + \frac{\log(q)}{q^2} \Big),
		\]
completing the proof.
	\end{proof}

\section*{Acknowledgments}
The authors thank Alon Rosen and Prashant Vasudevan for helpful comments on a previous version of this paper.

	
	\bibliographystyle{siam}
	\bibliography{refs3}

\appendix

\section{Wagner's \tops{$k$}-tree Algorithm}
\label{app:ktree}

In this appendix we sketch the details of Wagner's $k$-tree algorithm for solving the $k$-XOR problem and its generalization published in~\cite{MS12}.
The variant for solving $k$-SUM is similar. For more details, we refer the reader to the original publications~\cite{MS12,W02}.

\subsection{The 4-XOR algorithm}
We begin by describing the algorithm applied to a 4-list variant of 4-XOR.
In this problem, the input consists of 4 lists $\{x^{(j)}\}_{j=1}^{4}$,
where each $x^{(j)} \in \{0,1\}^{2^{n/3} \times n}$ is chosen uniformly at random.
The goal is to find 4 vectors, one from each list, whose XOR is $0_n$,
namely, output a 4-tuple $\{i_j\}_{j=1}^{4}$, where $i_j \in [2^{n/3}]$
such that $\bigoplus_{j=1}^{4} x^{(j)}_{i_j} = 0_n$.
It is easy to see the 4-list variant is equivalent to the single-list variant
(Definition~\ref{def:kxor})
up to $O(1)$ factors in success probability and complexity.

The $k$-tree algorithm for $k=4$ is described below.
\begin{enumerate}
\item Sort the lists $\{x^{(j)}\}_{j=1}^{4}$.
\item By a linear scan, find all pairs $(x^{(1)}_{i_1},x^{(2)}_{i_2})$ such that the $n/3$ most significant bits of $x^{(1)}_{i_1} \oplus x^{(2)}_{i_2}$ are zero. Store all values $x^{(1)}_{i_1} \oplus x^{(2)}_{i_2}$ in a new sorted list $y^{(1)}$, along with the corresponding pair $(x^{(1)}_{i_1},x^{(2)}_{i_2})$.
\item Apply the previous step to $x^{(3)}$ and $x^{(4)}$ and build the sorted list $y^{(2)}$.
\item Find a pair $(y^{(1)}_{j_1},y^{(2)}_{j_2})$ such that $y^{(1)}_{j_1} \oplus y^{(2)}_{j_2} = 0_n$. Trace $(y^{(1)}_{j_1},y^{(2)}_{j_2})$ back to a solution to 4-XOR problem and output it.
\end{enumerate}

To analyze the algorithm, note that the expected size of $y^{(1)}$ and $y^{(2)}$ is $2^{n/3}$ (as a pair $(x^{(1)}_{i_1},x^{(2)}_{i_2})$ is added to $y^{(1)}$ with probability $2^{-n/3}$).
Therefore, the algorithm runs in expected time $\tilde{O}(2^{n/3})$.
Moreover, on average, there is a single 4-XOR solution to be found in the last step, since any 4-tuple
$\{x^{(j)}_{i_j}\}_{j=1}^{4}$ satisfies the $4n/3$ bit constraints imposed by the algorithm with probability $2^{-4n/3}$ (and there are $2^{4n/3}$ such 4-tuples).
Based on tail bounds, one can show that the algorithm succeeds with constant probability. We refer the reader to~\cite{MS12} for a rigorous analysis.

\subsection{Generalizations}
We briefly summarize two important generalizations of the 4-XOR algorithm.

\subsubsection{\tops{The full $k$-tree algorithm~\cite{W02}}} The first generalization
applies to larger $k$ that is a power of 2.
The input consists of $k$ lists, each containing $2^{n/(\log k + 1)}$ vectors of $n$ bits. The algorithm merges the $k$ lists in pairs in a tree-like structure with $\log k + 1$ levels.
The merging maintains the property that the vectors in all $k/2^{\ell}$ lists in level $\ell \in \{0,1,\ldots,\log k -1\}$ have zero $\ell \cdot n/(\log k + 1)$ most significant bits.
The final 2-list merge at level $\log k - 1$ zeroes the remaining $2n/(\log k + 1)$ bits, giving a $k$-XOR solution at the last level with high probability.

When $k$ is not a power of 2, the $k$-XOR problem can be easily reduced to a $k'$-XOR problem where
$k'$ is the largest power of 2 that is smaller than $k$.

\subsubsection{\tops{The extended $k$-tree algorithm~\cite{MS12}}} This generalization applies when the input lists contain less than $2^{n/(\log k + 1)}$ vectors (i.e., the input is less dense) and the $k$-tree algorithm is not directly applicable. The extended algorithm gives a tradeoff between the size of
the inputs lists and the time complexity.

Specifically, for $4$-XOR, when the input lists are of size $r$ for $2^{n/4} \leq r \leq 2^{n/3}$,
we change the second step to find all pairs $(x^{(1)}_{i_1},x^{(2)}_{i_2})$ such that the
$4 \log r - n$ most significant bits of $x^{(1)}_{i_1} \oplus x^{(2)}_{i_2}$ are equal to $0_{4 \log r - 1}$ (we also change the third step similarly).
Therefore, the expected size of $y^{(1)}$ and $y^{(2)}$ becomes $2^{n}/r^2$, and the expected complexity of the algorithm is $\tilde{O}(2^{n}/r^2)$.
Finally, on average, there is a single 4-XOR solution to be found in the last step, since there are
$4 \log r$ bit constraints imposed by the algorithm on $r^4$ 4-tuples (once again, a tail bound is required to rigorously compute the success probability).

\end{document}